\documentclass[letterpaper,12pt]{article}
\pdfoutput=1

\usepackage[T1]{fontenc}
\usepackage[top=1in,left=1in,right=1in,bottom=1in]{geometry}
\usepackage{ragged2e}
\usepackage{lmodern}
\usepackage{multicol}
\usepackage{booktabs}
\usepackage{amsmath}
\usepackage{amsthm}
\usepackage{amsfonts}
\usepackage{thmtools}
\usepackage{mathtools}
\usepackage{amssymb}
\usepackage{bm}
\usepackage{enumitem}
\usepackage{graphicx}
\usepackage[dvipsnames]{xcolor}
\usepackage[bottom]{footmisc}
\usepackage{hyperref}
\usepackage{caption}
\usepackage{subcaption}
\usepackage{comment}
\usepackage[textsize=tiny]{todonotes}
\usepackage[normalem]{ulem}
\usepackage[nodisplayskipstretch]{setspace}

\allowdisplaybreaks

\definecolor{MATblue}{HTML}{0072BD}
\definecolor{MATpurple}{HTML}{7E2F8E}
\hypersetup{
	colorlinks=true,
	linkcolor=MATblue,
	citecolor=MATblue,
	urlcolor=MATblue
}

\usepackage[authoryear,round,comma]{natbib}
\newcommand{\bof}[1]{\bm{#1}}
\newcommand{\E}{\mathbb{E}}
\newcommand{\Prob}{{\mathbb{P}}}

\newcommand{\inv}{^{-1}}
\newcommand{\abs}[1]{\left\lvert#1\right\rvert}

\newcommand{\norm}[1]{\left\lVert#1\right\rVert}

\newcommand{\Real}{{\mathbb{R}}}
\newcommand{\Complex}{{\mathbb{C}}}
\newcommand{\vect}{{\textnormal{vec}}}
\newcommand{\diag}{{\textnormal{diag}}}
\newcommand{\trace}{{\textnormal{tr}}}

\newcommand{\toprob}{\overset{P}{\to}}
\newcommand{\todist}{\overset{d}{\to}}
\newcommand{\argmin}[1]{\textnormal{arg}\min_{#1}}

\theoremstyle{definition}

\theoremstyle{plain}
\newtheorem{theorem}{Theorem}
\newtheorem{corollary}{Corollary}
\newtheorem{prop}{Proposition}
\newtheorem{lemma}{Lemma}
\newtheorem{lemma_apx}{Lemma}[section]
\theoremstyle{definition}

\declaretheorem[name=Remark]{remark}

\begin{document}
	
\title{\normalfont Ridge Regularized Estimation of VAR\\ Models for Inference}
\author{%
	Giovanni Ballarin\thanks{%
		Department of Economics, University of Mannheim, L7, 3-5, Mannheim, 68131, Germany. \texttt{Giovanni.Ballarin@gess.uni-mannheim.de}.
		I am grateful for the comments and suggestions from Lyudmila Grigoryeva, So Jin Lee, Thomasz Olma, Oliver Pf{\"a}uti and Mikkel Plagborg-M{\o}ller, and the seminar participants at the University of Mannheim, the HKMetrics Workshop and the Young Researchers Workshop on Big and Smart Data Analysis in Finance. 
		I am especially thankful to Claudia Noack for pointing out an important error in a previous version of this paper, as well as Jonas Krampe and Carsten Trenkler for their insightful discussions which helped develop this work significantly.
		Lastly, I wish to thank Peter C. B. Phillips, Atsushi Inoue and many other colleagues for the suggestion to consider adding a formal analysis of cross-validation in the paper.
		The author acknowledges support by the state of Baden-Württemberg through bwHPC.
	} \\
	University of Mannheim%
}
\date{\today}
	
\maketitle

\renewcommand{\abstractname}{\vspace{-\baselineskip}} 
\begin{abstract}
	\textbf{Abstract:}
	Ridge regression is a popular method for dense least squares regularization. In this work, ridge regression is studied in the context of VAR model estimation and inference. The implications of anisotropic penalization are discussed and a comparison is made with Bayesian ridge-type estimators. The asymptotic distribution and the properties of cross-validation techniques are analyzed.  Finally, the estimation of impulse response functions is evaluated with Monte Carlo simulations and ridge regression is compared with a number of similar and competing methods. 
\end{abstract}


\pagebreak

\section{Introduction}

While the idea of using ridge regression for vector autoregressive model estimation dates back to \cite{hamilton1994time}, there seems to be no complete analysis of its properties and asymptotic theory in the literature. This paper fills this gap by analyzing the geometric and distributional properties of ridge in a VAR estimation framework, discussing its comparison to well-known Bayesian approaches and deriving the validity of cross-validation as a selection procedure for the ridge penalty. 

First, I show that the shrinkage induced by the ridge estimator, while intuitive in the setting of an isotropic penalty, produces complex effects when estimating a VAR model with a more flexible penalization scheme. This implies that the benefits of the bias-variance trade-off \citep{hastieRidgeRegularizationEssential2020} may be be hard to gauge a priori. I provide a tractable example where ridge can yield estimates that have {higher autoregressive dependence} than the least squares solution. To better understand how different ridge penalization strategies can be designed, I also make a comparison with Bayesian VAR estimators commonly used in macroeconometric practice. 

Second, I generalize the analysis of \cite{fuAsymptoticsLassotypeEstimators2000} and prove the consistency and asymptotic normality of the ridge estimator, a result that seems to be missing in the literature. For standard inference, the ridge penalty should either be negligible in the limit or its centering converge in probability to the true parameter vector. In both these cases, there is no asymptotic bias and no reduction in variance. Alternatively, in settings where a researcher is willing to assume that a subset of the VAR parameters features small coefficients, one can achieve an asymptotic reduction of variance by correctly tuning the ridge penalty matrix. I further derive the properties of cross-validation, which is a popular approach in practical applications to tune penalized estimators \citep{hastie2009elements,bergmeirNoteValidityCrossvalidation2018}. More specifically, I show that cross-validation is able to select penalties that are asymptotically valid for inference. In passing, I also prove that in an autoregressive setup the time dependence of regressors has an exponentially small effect on in-sample prediction error evaluation.

Lastly, I use Monte Carlo simulations to study the performance of the different ridge approaches discussed, focusing on impulse response inference. I consider two exercises: one is based on a three-variable VARMA(1,1) data generating process from \cite{kilianHowReliableAre2011}; the other is a VAR(5) model estimated in levels from a set of seven macroeconomic series, following \cite{giannonePriorSelectionVector2015}. The finding is that ridge can lead to improvements over unregularized methods in impulse response confidence interval length, while Bayesian estimators show the best overall performance due to the underlying flexibility of their priors.

\paragraph{Related Literature.}

This paper does not discuss the high-dimensional setting, where the number of regressors grows together with the sample size. Some important work has been done in this direction already. \cite{Dobriban2018} derive an explicit expression for the predictive risk of ridge regression assuming a high-dimensional random effects model. Other works in this vein are \cite{Liu2020Ridge,Patil2021} and \cite{Hastie2022}, which are mostly focused on penalty selection by cross-validation as well as structural features of ridge. Generally speaking, the complexity of analyzing ridge regression in high dimensions is a challenge to precisely understanding its practical implications. As I show below, in the context of finite-dimensional VARs, asymptotic inference demands that the ridge penalty becomes asymptotically negligible at appropriate rates. Thus, a challenge is understanding in what way high-dimensional time series problems would benefit from the use of ridge. This question is beyond the scope of this paper.

In the time series forecasting literature, ridge regression is commonly used for prediction. I provide a partial list of contributions in this direction. \cite{inoueHowUsefulBagging2008} use ridge regularization for forecasting U.S. consumer price inflation and argue that it compares favorably with bagging techniques; \cite{demolForecastingUsingLarge2008a} use a Bayesian VAR with posterior mean equivalent to a ridge regression in forecasting; \cite{ghoshHighDimensionalPosteriorConsistency2019} again study the Bayesian ridge, this time however in the high-dimensional context; \cite{coulombeHowMachineLearning2020}, \cite{Fuleky2020}, \cite{babiiMachineLearningTime2021}, and \cite{medeirosForecastingInflationDataRich2021} compare LASSO, ridge and other machine learning techniques for forecasting with large economic datasets. \cite{Fuleky2020} gives a textbook treatment of penalized time series estimation, including ridge, but does not discuss inference. The ridge penalty is considered within a more general mixed $\ell_1$-$\ell_2$ penalization setting in \cite{smeekesMacroeconomicForecastingUsing2018}, who study the performance and robustness of penalized estimates for constructing forecasts. 

Regarding inference, \cite{liLocalProjectionsVs2022} provided a general exploration of shrinkage procedures in the context of structural impulse response estimation. Very recently, \cite{cavaliereBootstrapInferencePresence2022} suggested a methodology for inference on ridge-type estimators that relies on bootstrapping. Finally, shrinkage of autoregressive models to constrained sub-models was discussed by \cite{hansenSteinCombinationShrinkage2016a} in a more general setting.

Finally, various estimation problems can either be cast as or augmented with ridge-type regressions. \cite{GouletCoulombe2023} shows that the estimation of VARs with time-varying parameters can be written as ridge regression. \cite{plagborg2016essays} and \cite{barnichonImpulseResponseEstimation2019} both use ridge to derive smoothed local projection impulse response functions.

\paragraph{Outline.}

Section \ref{section:ridge_var} provides a discussion of the ridge penalty and the ridge VAR estimator. In Section \ref{section:ridge_shrinkage} I deal with the properties of ridge-induced shrinkage in the autoregressive coefficients. I discuss the connections between frequentist and Bayesian ridge for VAR models within Section \ref{section:bayesian_ridge}. Section \ref{section:standard_inference} develops the asymptotic theory and inference result in the case where there is no asymptotic shrinkage. This includes studying the property of cross-validation under dependence. Section \ref{section:inference_shrinkage} provides inference and CV results in a setting where some shrinkage of a subset of parameters is possible. Section \ref{section:monte_carlo_simulations} presents Monte Carlo simulations focused on impulse response estimation. Section \ref{sect:end} concludes. Finally, the Appendix contains more detailed derivations, as well as all proofs and supplementary information.

\paragraph{Notation.}

Define $\Real_{+}$ to be the set of strictly positive real numbers. Vectors $v \in \Real^N$ and matrices $A \in \Real^{N\times M}$ are always denoted with lower and upper-case letters, respectively. Throughtout, I will use $I_M$ to represent the identity matrix of dimension $M$. For any vector $v \in \Real^N$, $\norm{v}$ is the Euclidean norm. For any matrix $A \in \Real^{N \times M}$, $\norm{A}$ is the spectral norm unless stated otherwise; $\norm{A}_{\max} = \max_{i,j} \abs{a_{ij}}$ is the maximal entry norm; $\norm{A}_F = (\trace\{A' A\})^{-1/2}$ is the Frobenius norm; $\vect(\cdot)$ is the vectorization operator and $\otimes$ is the Kronecker product \citep{lutkepohl2005new}. If a vector represents a vectorized matrix, then it will be written in bold, that is, for $A \in \Real^{N\times M}$ I write $\vect(A) = \bof{a} \in \Real^{NM}$. 
Let $\Lambda = \diag\{\lambda_1, \ldots, \lambda_{K^2 p}\}$, $\lambda_i  > 0 $ for all $i = 1, \ldots, K^2 p$. To give the partial ordering of diagonal positive semi-definite penalization matrices, let $\Lambda_1 = \diag\{\lambda_{1,j}\}_{j=1}^{K^2 p}$ and $\Lambda_2 = \diag\{\lambda_{2,j}\}_{j=1}^{K^2 p}$. I write $\Lambda_1 \prec \Lambda_2$ if $\lambda_{1,i} < \lambda_{2,i}$ for all $i = 1, \ldots, K^2p$; $\Lambda_1 \preceq \Lambda_2$ if $\lambda_{1,i} \leq \lambda_{2,i}$  for all $i$ and $\exists\ j \in {1, \ldots, K^2p}$ such that $\lambda_{1,j} < \lambda_{2,j}$.
Symbols $\toprob$ and $\todist$ are used to indicate convergence in probability and convergence in distribution, respectively.

\section{Ridge Regularized VAR Estimation}\label{section:ridge_var}

Let $y_t = (y_{1t}, \ldots, y_{Kt})'$ be a $K$-dimensional vector autoregressive process with lag length $p \geq 1$ and parametrization 
\begin{equation}\label{eq:linear_var_model}
	y_t = \nu_t + A_1 y_{t-1} + A_2 y_{t-2} + \ldots + A_p y_{t-p} + u_t ,
\end{equation}
where $u_t = (u_{1t}, \ldots, u_{Kt})'$ is additive noise such that $u_t$ are identically, independently distributed with $\E[u_{it}] = 0$ and $\text{Var}[u_t] = \Sigma_u$, and $\nu_t$ is a deterministic trend. For simplicity, in the remainder I shall assume that $\nu_t = 0$ so that $y_t$ has no trend component -- equivalently, $y_t$ is a de-trended series. 

For a given sample size $T$ define $Y = (y_1, \ldots, y_T) \in \Real^{K\times T}$, $z_t = (y_t', y_{t-1}', \ldots, y_{t-p+1}')' \in \Real^{Kp}$, $Z = (z_0, \dots, z_{T-1}) \in \Real^{Kp\times T}$, $B = (A_1, \ldots, A_p) \in \Real^{K\times Kp}$, $U = (u_1, \ldots, u_T) \in \Real^{K\times T}$, and vectorized counterparts $\bof{y} = \vect(Y)$, $\bm{\beta} = \vect(B)$ and $\bof{u} = \vect(U)$. Accordingly, $Y = B Z + U$ and $\bof{y} = (Z' \otimes I_K) \bm{\beta} + \bof{u}$, where $\Sigma_{\bof{u}} = I_K \otimes \Sigma_u$. Importantly, throughout this work, I will assume that the cross-sectional dimension $K$ remains fixed.

Ridge regularization is a modification of the least squares objective by the addition of a term dependent on the Euclidean norm of the coefficient vector. The \textit{isotropic} {Ridge-regularized Least Squares} (RLS) estimator is therefore defined as
\begin{equation*}
	\hat{\bm{\beta}}^R(\lambda) 
	:= 
	 \argmin{\bm{\beta}}\, \frac{1}{T} \norm{\bm{y} - (Z' \otimes I_K)\bm{\beta}}^2 + \lambda \norm{\bm{\beta}}^2 ,
\end{equation*}
where $\lambda > 0$ is the scalar regularization parameter or regularizer. When $\lambda \norm{\bm{\beta}}^2$ is replaced with quadratic form $\bm{\beta}' \Lambda \bm{\beta}$ for a positive definite matrix $\Lambda$ the above is often termed Tikhonov regularization. To avoid confusion, I shall also refer to it as ``ridge'', since in what follows $\Lambda$ will always be assumed diagonal. Since $\Lambda$ does not, in general, penalize coefficients equally, it will yield an \textit{anisotropic} ridge estimator. 

By solving the normal equations (see Appendix \ref{Supp_Appendix:LS_and_RLS_Estimators}), the RLS estimator with positive semi-definite regularization matrix $\Lambda \in \Real^{K^2 p \times K^2 p}$ is shown to be
\begin{equation*}\label{eq:RLS}
	\hat{\bm{\beta}}^R(\Lambda) = \left( \frac{Z Z'}{T} \otimes I_K + \Lambda \right)\inv \frac{(Z \otimes I_K)\bof{y}}{T}.
\end{equation*}
When a centering vector $\bm{\beta}_0 \not= 0$ is included in penalty $(\bm{\beta} - \bm{\beta}_0)' \Lambda (\bm{\beta} - \bm{\beta}_0)$, the RLS estimator becomes
\begin{equation}\label{eq:RLS_b0}
	\hat{\bm{\beta}}^R(\Lambda, \bm{\beta}_0) = \left( \frac{Z Z'}{T} \otimes I_K + \Lambda \right)\inv \left( \frac{(Z \otimes I_K)\bof{y}}{T} + \Lambda \bm{\beta}_0 \right).
\end{equation}

In the context of multivariate estimation, one has to make a further distinction between two related types of ridge estimators.
I let $\hat{B}^R(\Lambda, \bm{\beta}_0)$ be the {de-vectorized} coefficient estimator obtained from reshaping $\hat{\bm{\beta}}^R(\Lambda, \bm{\beta}_0)$ to a $K \times K p$ matrix. But one can also consider the \textit{matrix RLS estimator} $\hat{B}^R_{\textnormal{mat}}(\Lambda_{Kp}, B_0)$ given by
\begin{equation*}
	\hat{B}^R_{\textnormal{mat}}(\Lambda_{Kp}, B_0) = T\inv (Y + B_0\, \Lambda_{Kp}) Z' \left( T\inv Z Z' + \Lambda_{Kp} \right)\inv ,
\end{equation*}
where $\Lambda_{Kp} = \diag\{\lambda_1, \ldots, \lambda_{K p}\}$ and $B_0$ is a centering matrix. The distinction is important because the vectorized and matrix RLS estimators in general need not coincide. As discussed in Appendix~\ref{Supp_Appendix:Structure_Lambda}, $\hat{B}^R(\Lambda, \bm{\beta}_0)$ allows for more general penalty structures compared to $\hat{B}^R_{\textnormal{mat}}(\Lambda_{Kp}, B_0)$. I, therefore, focus on the former rather than the latter.

\section{Shrinkage}\label{section:ridge_shrinkage}

In this section, I discuss both the isotropic ridge penalty, i.e. the ``standard'' ridge approach, as well as an anisotropic penalty that is better adapted to the VAR setting. An important result is that, even in simple setups with only two variables, the shrinkage induced by ridge can either increase or reduce bias, as well as the stability of autoregressive estimates.

Throughout this section, I consider \textit{fixed} design matrices and the focus will be on the geometric properties of ridge.

\subsection{Isotropic Penalty}

The most common way to perform a ridge regression is through {isotropic} regularization, that is, $\Lambda = \lambda I$ for some scalar $\lambda \geq 0$. Isotropic ridge has been extensively studied, see for example the comprehensive review of \cite{hastieRidgeRegularizationEssential2020}. In regards to shrinkage, an isotropic ridge penalty can be readily studied.

\begin{prop}\label{prop:iso_shrink}
	Let $Z \in \Real^{M\times T}$, $Y \in \Real^{T}$ for $T > M$ be regression matrices. For $\lambda_\bullet > \lambda > 0$ and isotropic RLS estimator $\hat{\beta}^R(\lambda) := (T\inv Z Z' + \lambda I_M)\inv (T\inv Z Y)$ it holds
	\begin{equation*}
		\norm{\hat{\beta}^R(\lambda_\bullet)} < \norm{\hat{\beta}^R(\lambda)}.
	\end{equation*}
\end{prop}

\begin{proof}
	Using the full singular-value decomposition (SVD), decompose $T^{-1/2} Z=U D V' \in \Real^{M\times T}$ where $U$ is $M\times M$ orthogonal, $D$ is $M \times T$ diagonal and $V$ is $T \times T$ orthogonal. Write
	\begin{align*}
		\hat{\beta}^R(\lambda_\bullet) 
		& = (T\inv Z Z' + \lambda_\bullet I_M)\inv (T\inv Z Y) \\
		& = (U D V' V D U' + \lambda_\bullet I_M)\inv U D V' (T^{-1/2} Y) \\
		& = U (D^2 + \lambda_\bullet I_M)\inv D V' (T^{-1/2} Y) \\
		& = U (D^2 + \lambda_\bullet I_M)\inv (D^2 + \lambda I_M) (D^2 + \lambda I_M)\inv D V' (T^{-1/2} Y)\\
		& = \left[ U (D^2 + \lambda_\bullet I_M)\inv (D^2 + \lambda I_M) U' \right] \hat{\beta}^R(\lambda).
	\end{align*}
	Since $D^2 = \diag\{\sigma_j^2\}_{j=1}^M$, the term within brackets is $U\, \diag\{ (\sigma_j^2 + \lambda) / (\sigma_j^2 + \lambda_\bullet) \}_{j=1}^M\, U'$. Moreover, because the spectral norm is unitary invariant and $\lambda_1 > \lambda_2$, it follows that
	\begin{equation*}
		\norm{U (D^2 + \lambda_\bullet I_M)\inv (D^2 + \lambda I_M) U'} 
		= 
		\norm{ \diag\{ (\sigma_j^2 + \lambda) / (\sigma_j^2 + \lambda_\bullet) \}_{j=1}^M } < 1 .
	\end{equation*}
	Finally, by the sub-multiplicative property it holds
	\begin{equation*}
		\norm{\hat{\beta}^R(\lambda_\bullet)} 
		\leq 
		\norm{U (D^2 + \lambda_1 I_M)\inv (D^2 + \lambda I_M) U'} \cdot \norm{\hat{\beta}^R(\lambda)} < \norm{\hat{\beta}^R(\lambda)} 
	\end{equation*}
	as claimed.
\end{proof}

Proposition \ref{prop:iso_shrink} and its proof expose the main ingredients of ridge regression. From the SVD of $Z$ used above, it is clear that ridge regularization acts uniformly along the orthogonal directions that are the columns of $V$. The improvement in conditioning of the inverse comes from all diagonal factors $[(D^2 + \lambda_\bullet I_M)\inv D]_j = \sigma_j / (\sigma_j^2 + \lambda_\bullet)$ being well-defined even when $\sigma_j = 0$ (as is the case in collinear systems). 

However, directly applying isotropic ridge to vector autoregressive models is not necessarily the most effective estimation approach. Stable VAR models show decay in the absolute size of coefficients over lags. So it is reasonable to chose a more general ridge penalty that can accommodate lag decay.

\subsection{Lag-Adapted Penalty}

I now consider a different form for $\Lambda$ that is of interest when applying ridge specifically to a VAR model.
Define family $\mathcal{F}^{(p)}$ of \textit{lag-adapted} ridge penalty matrices as
\begin{equation*}
	\mathcal{F}^{(p)} = \{\diag\{\lambda_1, \ldots, \lambda_p\} \otimes I_{K^2} \:\vert\: \lambda_i \in \Real_{+},\, i = 1, \ldots, p\},
\end{equation*}
where each $\lambda_i$ intuitively implies a different penalty for the elements of each coefficient matrix $A_i$, $i = 1, \ldots, p$.\footnote{Note that with a lag-adapted penalty it is also possible to directly use the matrix ridge estimator since the penalty for $\hat{\bm{\beta}}^R$ is given by $\diag\{\lambda_1, \ldots, \lambda_p\} \otimes I_{K^2} = (\diag\{\lambda_1, \ldots, \lambda_p\} \otimes I_K) \otimes I_K$, see Appendix \ref{Supp_Appendix:Structure_Lambda}.} Family $\mathcal{F}^{(p)}$ allows imposing a ridge penalty that is coherent with the lag dimension of an autoregressive model. It is parametrized by $p$ distinct penalty factors, meaning that the penalization is \textit{anisotropic}.

\begin{prop}\label{prop:lag_adapt_shrink}
	Let $Z \in \Real^{K p \times T}$, $\bof{y} \in \Real^{K T}$ for $T > K p$ be multivariate VAR regression matrices.
	Given subset $\mathcal{S} \subseteq \{1, \ldots, p\}$ of cardinality $s = \abs{\mathcal{S}}$, for $\Lambda^{(p)} \in \mathcal{F}^{(p)}$ define $\hat{\bm{\beta}}^R(\Lambda^{(p)})_{[\mathcal{S}]}$ as the vector of $s K^2$ coefficient estimates located at indexes $1 + K^2 (j-1), \ldots, K^2 j$ for $j \in \mathcal{S}$. 
	Let $\mathcal{S}^\mathnormal{c} = \{1, \ldots, p\} \setminus \mathcal{S}$ be the complement of $\mathcal{S}$. 
	\begin{itemize}
		\item[(a)] If $\lambda_1 \geq \lambda_2$, then $\norm{\hat{\bm{\beta}}^R(\lambda_1 I_{K^2p})_{[\mathcal{U}]}} \leq \norm{\hat{\bm{\beta}}^R(\lambda_2 I_{K^2p})_{[\mathcal{U}]}}$ for any $\mathcal{U} \subset \{1, \ldots, K^2 p\}$. The inequality is strict when $\lambda_1 > \lambda_2$. 
		\item[(b)] Let $\hat{\bm{\beta}}^{LS}_{[\mathcal{S}]}$ be the least squares estimator of the autoregressive model with only the lags indexed by $\mathcal{S}$ included and zeros as coefficients for the lags indexed by $\mathcal{S}^\mathnormal{c}$. Similarly, let $\Lambda^{(p)}_{[\mathcal{S}]}$ be the subset of diagonal elements in $\Lambda^{(p)}$ penalizing the lags in $\mathcal{S}$. Then
		\begin{equation*}
			\lim_{\substack{\Lambda^{(p)}_{[\mathcal{S}]} \to 0 \\ \Lambda^{(p)}_{[\mathcal{S}^\mathnormal{c}]} \to \infty}} \hat{\bm{\beta}}^R(\Lambda^{(p)}) 
            = 
            \hat{\bm{\beta}}^{LS}_{[\mathcal{S}]}, 
		\end{equation*}
        where $\Lambda^{(p)}_{[\mathcal{S}]} \to 0$ and $\Lambda^{(p)}_{[\mathcal{S}^\mathnormal{c}]} \to \infty$ are to be intended as the element-wise convergence.
	\end{itemize}
\end{prop}

Proposition \ref{prop:lag_adapt_shrink} shows that the limiting geometry of a lag-adapted ridge estimator is thus identical to that of a least squares regression run on the subset specified by $\mathcal{S}$. By controlling the size of coefficients $\{\lambda_1, \ldots, \lambda_p\}$ it is therefore possible to obtain pseudo-model-selection. However, in the next section I show that anisotropic penalization produces complex effects on the models' coefficient estimates. 

\subsection{Effects of Anisotropic Penalization}

In this section, I explore the effect of a lag-adapted ridge penalty on the estimate VAR coefficients, and, more generally, the properties of ridge estimators with anisotropic penalization.

Since ridge operates along the principal components, there is no immediate relationship between a specific subset of estimated coefficients and a given diagonal block of $\Lambda^{(p)}$. For autoregressive modeling, three effects are of interest: the shrinkage of coefficient matrices $A_i$ relative to the choice of $\lambda_i$; the entity of the bias introduced by shrinkage, and the impact of shrinkage on the persistence of the estimated model. 

To evaluate these effects, I consider a simple VAR(2) model
\begin{equation*}
	y_t = A_1 y_{t-1} + A_2 y_{t-2} + u_t, \quad u_t \sim \text{i.i.d.}\ \mathcal{N}(0, \Sigma_u),
\end{equation*}
where
\begin{equation*}
	A_1 = \begin{bmatrix}
		0.8 & 0.1 \\
		-0.1 & 0.7
	\end{bmatrix},
	\quad
	A_2 = \begin{bmatrix}
		0.1 & -0.2 \\
		-0.1 & 0.1
	\end{bmatrix},
	\quad 
	\Sigma_u = \begin{bmatrix}
		0.3 &   0 \\  
		0  &   5
	\end{bmatrix}
	.
\end{equation*}
~ \\
A sample of length $T = 200$ is drawn, demeaned and used to estimate coefficients ${A}_1$ and $A_2$. The VAR(2) model is fitted using the lag-adapted ridge estimator $\hat{B}^R(\Lambda^{(2)})$, where $\Lambda^{(2)} = \diag\{\lambda_1, \lambda_2\} \otimes I_2$, which can be partitioned into estimates $\hat{A}_1^R(\Lambda^{(2)})$ and $\hat{A}_2^R(\Lambda^{(2)})$.

\paragraph{Shrinkage.}
To study shrinkage, I consider the restricted case of $\lambda_1 \in [10^{-2}, 10^{6}]$ and $\lambda_2 = 0$. The ridge estimator is computed for varying $\lambda_1$ over a logarithmically spaced grid. Figure \ref{fig:ridgeshrinknorm__A} shows that $\norm{\hat{B}^R(\Lambda^{(2)})}_F \approx \norm{\hat{B}^{LS}}_F$ for $\lambda_1 \approx 0$, but as the penalty increases $\norm{\hat{A}_1^R(\Lambda^{(2)})}_F$ decreases while $\norm{\hat{A}_2^R(\Lambda^{(2)})}_F$ grows. The resulting behavior of $\norm{\hat{B}^R(\Lambda^{(2)})}_F$ is non-monotonic in $\lambda_1$, although indeed $\norm{\hat{B}^R(\Lambda^{(2)})}_F < \norm{\hat{B}^{LS}}_F$ in the limit $\lambda_1 \to \infty$. 
This effect is due to the model selection properties of lag-adapted ridge, and the resulting omitted variable bias. Therefore, in practice it is not generally true that anisotropic ridge induces monotonic shrinkage of estimates.

\begin{figure}[t!]
	\centering
	\makebox[\textwidth][c]{
		\begin{subfigure}{0.49\textwidth}
			\includegraphics[width=\textwidth]{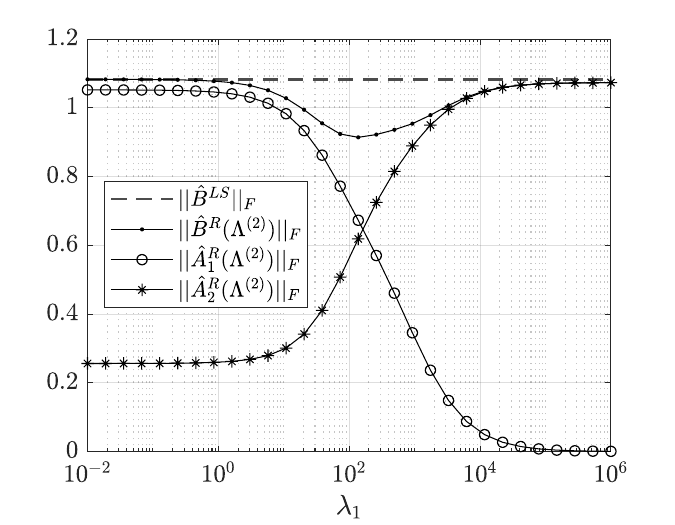}
			\caption{}
			\label{fig:ridgeshrinknorm__A}
		\end{subfigure}
		\begin{subfigure}{0.49\textwidth}
			\includegraphics[width=\textwidth]{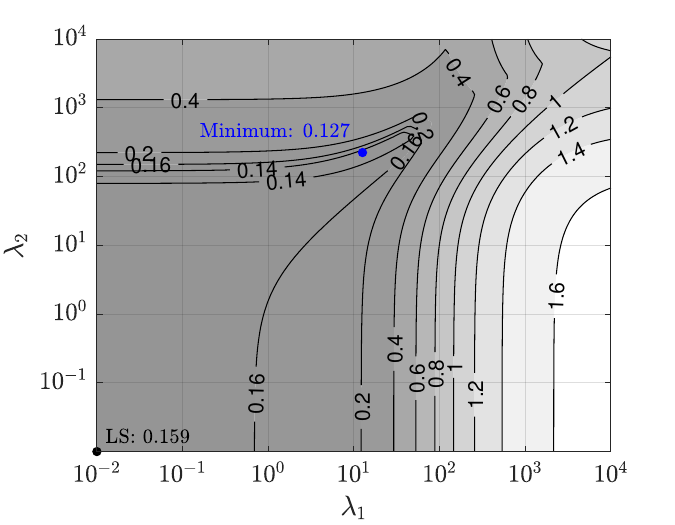}
			\caption{}
			\label{fig:ridgeshrinknorm__B}
		\end{subfigure}
	}\\
	\vspace{.5em}
	\begin{subfigure}{0.49\textwidth}
		\includegraphics[width=\textwidth]{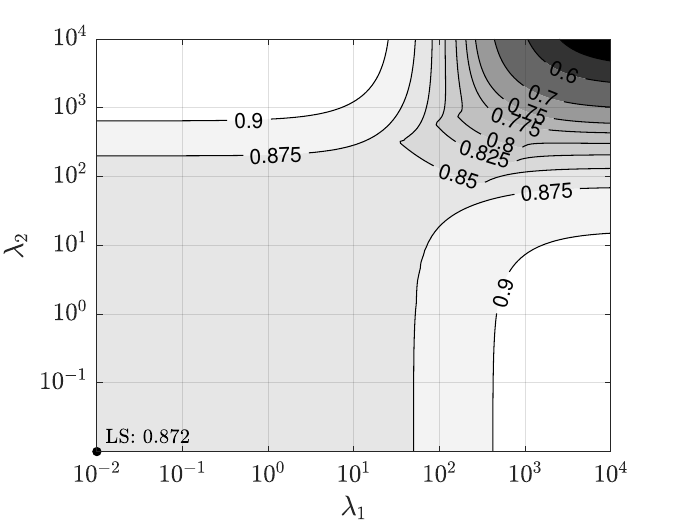}
		\caption{}
		\label{fig:ridgeshrinknorm__C}
	\end{subfigure}
	\vspace{0.3em}
	\caption{Shrinkage of coefficients estimate in Frobenius norm (a); bias induced by shrinkage (b); change in stability of estimated VAR model at different levels of penalization, measured by the absolute value of the largest companion form eigenvalue (c).}
	\label{fig:ridgeshrinknorm}
\end{figure}

\paragraph{Bias.}
Since ridge bias is hard to study theoretically, I use a simulation with the same setup of Figure \ref{fig:ridgeshrinknorm__A}, this time with $\lambda_1, \lambda_2 \in [10^{-2}, 10^{4}]$. The grid is logarithmic with $150$ points.
Figure \ref{fig:ridgeshrinknorm__B} presents a level plot of the sup-norm ridge bias $\norm{ \hat{B}^R(\Lambda^{(2)}) - B }_\infty$ given multiple combinations of $\lambda_1$ and $\lambda_2$. While there can be gains compared to the least squares estimator $\hat{B}^{LS}$, they are modest. Moreover, level curves of the bias surface show that gains concentrate in a very thin region of the parameter space. Consequently, in practice any (data-driven) ridge penalty selection criterion is unlikely to yield bias improvement over least squares. Yet, in large VAR models with many lags the reduction in variance of the ridge estimator often yields improvements over unregularized procedures \citep{liLocalProjectionsVs2022}.
However, the bias-variance trade-off in ridge is not a free-lunch when performing inference. \cite{prattLengthConfidenceIntervals1961} showed that it is not possible to produce a test (equivalently, a CI procedure) which is valid uniformly over the parameter space and yields meaningfully smaller confidence intervals than any other valid method. 

\paragraph{Stability.}
To study the stability of ridge VAR estimates, I reuse the results of the bias simulation above. Let $\mathbb{A}$ be the companion matrix of $[A_1, A_2]$, and $\hat{\mathbb{A}}^R$ the companion matrix of estimates $[\hat{A}_1^R(\Lambda^{(2)}), \hat{A}_2^R(\Lambda^{(2)})]$. For all combinations $(\lambda_1, \lambda_2)$, I compute the largest eigenvalue $\omega_1(\hat{\mathbb{A}}^R)$ of $\hat{\mathbb{A}}^R$. Note that if $\lvert \omega_1(\hat{\mathbb{A}}) \rvert < 1$, then the estimated VAR(2) is stable \citep{lutkepohl2005new}. Figure \ref{fig:ridgeshrinknorm__C} presents the level sets for the surface of maximal eigenvalue moduli, and for comparison $\lvert \omega_1(\hat{B}^{LS}) \rvert$ is shown at the origin.\footnote{If $\Lambda^{(2)} \to 0$, then by continuity of eigenvalues it follows that $\lvert \omega_1(\hat{\mathbb{A}}^R) \rvert \to \lvert \omega_1(\hat{\mathbb{A}}^{LS}) \rvert$, see Appendix \ref{Autocov_Matrix_Conditioning}.} While along the main diagonal there is a clear decrease in $\lvert \omega_1(\hat{\mathbb{A}}^R) \rvert$ as isotropic penalization increases, when $\lambda_1$ is large and $\lambda_2 \ll 1$ (or vice versa) the maximal eigenvalue increases instead. Therefore, an estimate of a VAR model obtained with anisotropic ridge may be \textit{closer} to unit root than the least squares estimate.

\section{Bayesian and Frequentist Ridge}\label{section:bayesian_ridge}

So far, I have discussed standard ridge penalization schemes. In this section, I study the posterior mean of Bayesian VAR (BVAR) priors commonly applied in the macroeconometrics literature. I show that such posteriors are in fact specific GLS formulations of the ridge estimator. This comparison highlights that ridge can be seen as a way to embed prior knowledge into the least squares estimation procedure by means of centering and rescaling coefficient estimates.

\subsection{Litterman-Minnesota Priors}

In Bayesian time series modeling, the so-called Minnesota or Litterman prior has found great success \citep{litterman1986}. For stationary processes which one believes to have reasonably small dependence, a zero-mean normal prior can be put on the VAR parameters, with non-zero prior variance. Assuming that the covariance matrix of errors $\Sigma_u$ is known, the Litterman-Minnesota has posterior mean
\begin{equation}\label{eq:litterman_posterior}
	\overline{\bm{\beta}} \:\vert\: \Sigma_u = \left[ \underline{V}_{\bm{\beta}}\inv + (Z Z' \otimes \Sigma_u\inv) \right]\inv (Z \otimes \Sigma_u\inv ) \bof{y},
\end{equation}
where $\underline{V}_{\bm{\beta}} \succ 0$ is the prior covariance matrix of $\bm{\beta}$ \citep{lutkepohl2005new}. It is common to let $\underline{V}_{\bm{\beta}}$ be diagonal, and often the entries follow a simple pattern which depends on lag, individual components variances, and prior hyperparameters. For example, \cite{banburaLargeBayesianVector2010} suggest the following structure for the diagonal
\begin{equation}\label{eq:litterman_posterior__var_prior}
	v_{i,jk} = \begin{cases}
		\dfrac{\lambda^2}{i^2} & \quad \text{if } j = k, \\
		\theta \dfrac{\lambda^2}{i^2} \dfrac{\sigma_j^2}{\sigma_k^2} & \quad \text{if } j \not= k,
	\end{cases} 
\end{equation}
where $v_{i,jk}$ is the prior variance for coefficients $(A_i)_{jk}$ for $i = 1, \ldots, p$ and $j,k = 1, \ldots, K$. Here, $\sigma_j$ is the $j$-th diagonal element of $\Sigma_u$, $\theta \in (0,1)$ specifies beliefs on the explanatory importance of own lags relative to other variables' lags, while $\lambda \in [0, \infty]$ controls the overall tightness of the prior. The extreme $\lambda = 0$ yields a degenerate prior centered at $\overline{\bm{\beta}} = 0$, while $\lambda = \infty$ reduces the posterior mean to OLS estimate $\hat{\bm{\beta}}^{LS}$. Factor $1/i^2$, which explicitly shrinks variance at higher lags, was originally introduced by \cite{demolForecastingUsingLarge2008a}, who formally developed the idea that coefficients at deeper lags should be coupled with more penalizing priors. Note that, in \eqref{eq:litterman_posterior__var_prior}, assuming $\sigma_j^2 = \sigma_k^2$ for all $j, k = 1, \ldots, K$ and setting $\theta = 1$, produces a $\underline{V}_{\bm{\beta}}$ that has a lag-adapted structure with quadratic lag decay. 

Equation \eqref{eq:litterman_posterior} more generally demonstrates that the Minnesota posterior mean is equivalent to a ridge procedure. It is important to notice that, while with least squares the OLS and GLS estimators of VAR coefficients coincide, this is not the case with ridge regression. Regularizing a GLS regression will yield
\begin{equation}\label{eq:GLS_RLS}
	\hat{\bm{\beta}}^{RGLS}(\Lambda) := \left[ \Lambda + (Z Z' \otimes \Sigma_u\inv) \right]\inv (Z \otimes \Sigma_u\inv ) \bof{y}.
\end{equation}
instead of $\hat{\bm{\beta}}^R$, which is equivalent to \eqref{eq:litterman_posterior} under an appropriate choice of $\Lambda$. While I develop the asymptotic results for $\hat{\bm{\beta}}^{R}$ assuming a centering parameter $\underline{\bm{\beta}}_0 \not= 0$ in general, I do not directly study the properties $\hat{\bm{\beta}}^{RGLS}$. The generalization to GLS ridge employing the least squares error covariance estimator $\hat{\Sigma}_T^{LS}$ should follow from straightforward arguments. In Section \ref{section:monte_carlo_simulations}, I focus on providing evidence on the application $\hat{\bm{\beta}}^{RGLS}$ in terms of its pointwise impulse response estimation mean-squared error.

\subsection{Hierarchical Priors}

Recent research on Bayesian vector autoregressions exploit more sophisticated priors compared to the Litterman-Minnesota design. \cite{giannonePriorSelectionVector2015} develop an advanced BVAR model by setting up hierarchical priors which entail not only model parameters, but also hyperparameters. They impose
\begin{align*}
	\Sigma_u \; & \sim \; \textnormal{IW}(\underline{\Psi}, \underline{d}) , \\
	\bm{\beta} \:\vert\: \Sigma_u \; & \sim \; \mathcal{N}\left( \underline{\bm{\beta}}, \lambda (\Sigma_u \otimes \underline{\Omega}) \right) ,
\end{align*}
for hyperparameters $\underline{\bm{\beta}}$, $\underline{\Omega}$, $\underline{\Psi}$ and $\underline{d}$, where $\textnormal{IW}$ is the Inverse-Wishart distribution. Here, too, scalar $\lambda \in [0, \infty]$ controls prior tightness. Let $\underline{B}$ be the matrix form of the VAR coefficient prior mean, so that $\vect(\underline{B}) = \underline{\bm{\beta}}$. The resulting (conditional) posterior mean $\overline{B}$ is given by
\begin{equation}\label{eq:glp_posterior}
	\overline{B} \:\vert\: \Sigma_u = \left[ (\lambda \underline{\Omega})\inv + Z Z' \right]\inv \left[ Z Y + (\lambda \underline{\Omega})\inv \underline{B} \right] .
\end{equation}
Observe that equation \eqref{eq:glp_posterior} is effectively equivalent to a centered ridge estimator, c.f. \eqref{eq:RLS_b0}. 

The introduction of a hierarchical prior leaves space to add informative hyperpriors on the model hyperparameters, allowing for a more flexible fit. Indeed, removing the zero centering constraint from the prior on $\bm{\beta}$ can improve estimation. It is often the case that economic time series show a high degree of correlation and temporal dependence, therefore imposing $\bm{\beta} = 0$ as in the Minnesota prior is inadequate. 
In fact, \cite{giannonePriorSelectionVector2015} show that their approach yields substantial improvements in forecasting exercises, even when hyperparameter priors are relatively flat and uninformative.

\section{Standard Inference}\label{section:standard_inference}

In this section, I state the main asymptotic results for the RLS estimator $\hat{\bm{\beta}}^R(\Lambda, \bm{\beta}_0)$ with general regularization matrix $\Lambda$. I shall allow $\Lambda$ and non-zero centering coefficient $\bm{\beta}_0$ to be, under appropriate assumptions, random variables dependent on sample size $T$. In particular, $\bm{\beta}_0$ may be a consistent estimator of $\bm{\beta}$. 

I will impose the following assumptions.

\vspace{-1em}
\paragraph{Assumptions}

\begin{itemize}\vspace{-0.5em}
	\item[A.] $\{u_t\}_{t=1}^T$ is a sequence of i.i.d. random variables with $\E[u_{it}] = 0$, covariance $\E[u_t u_t'] = \Sigma_u$ non-singular positive definite and $\E\abs{u_{it} u_{jt} u_{mt} u_{nt}} < \infty$, $i,j,m,n = 1, \ldots, K$. 
	\item[B.] There exists $\rho > 1$ such that $\det(I_K - \sum_{i=1}^{p} A_i z^i) \not= 0$ for all complex $z$, $\abs{z} \leq \rho$.
	\item[C.] There exist $0 < \underline{m} \leq \overline{m} < \infty$ such that $\underline{m} \leq \omega_{K}(\Gamma) \leq \omega_{1}(\Gamma) \leq \overline{m}$, where $\Gamma = \E[z_t z_t']$ is the autocovariance matrix of $z_t$ and $\omega_{1}(\Gamma)$, $\omega_{K}(\Gamma)$ are its largest and smallest eigenvalues, respectively.
\end{itemize}

Assumption A is standard and allows to prove the main asymptotic results with well-known theoretical devices. Assuming $u_t$ is white noise or assuming $y_t$ respects strong mixing conditions \citep{davidson1994stochastic} would require more careful consideration in asymptotic arguments but is otherwise a simple generalization, although more involved in terms of notation, see e.g. \cite{boubacarmainassaraEstimatingStructuralVARMA2011}.

Assumption B guarantees that $y_t$ has no unit roots and is stable. Of course, many setups of interest do not satisfy this assumption, the most significant ones being unit roots, cointegrated VARs, and local-to-unity settings. Incorrect identification of unit roots does not invalidate the use of LS or ML estimators \citep{phillipsRegressionTheoryNearIntegrated1988, park1988statistical, park1989statistical, sims1990inference}, however inference is significantly impacted as a result \citep{pesaventoSmallsampleConfidenceIntervals2006,mikushevaUniformInferenceAutoregressive2007,mikushevaOneDimensionalInferenceAutoregressive2012}. 

Assumption C is standard in the literature regarding penalized estimation and does not imply significant additional constraints on the process $y_t$, cf. Assumption A. It is sufficient to ensure that for large enough $T$ the plug-in sample autocovariance estimator is invertible even under vanishing $\Lambda$.

Before stating the main theorems, let
\begin{align*}
	\hat{\Gamma} &= T\inv Z Z', \\
	\hat{U} & = Y - \hat{B}^R Z, \\ 
	\hat{{\Sigma}}_u^R & = T\inv \hat{U} \hat{U}',
\end{align*}
be the regression covariance matrix, regression residuals and sample innovation covariance estimator, respectively.

\begin{theorem}\label{th:MAIN_THEOREM}
	Let Assumptions A-C hold and define $\hat{\bm{\beta}}^R(\Lambda, \bm{\beta}_0)$ be the centered RLS estimator as in \eqref{eq:RLS_b0}. 
	If $\sqrt{T} \Lambda \toprob \Lambda_0$ and $\bm{\beta}_0 \toprob \underline{\bm{\beta}}_0$, where $\Lambda_0$ is positive semi-definite diagonal matrix and $\underline{\bm{\beta}}_0$ is a constant vector, then
	\begin{itemize}
		\item[\textnormal{(a)}] $\hat{\Gamma} \toprob \Gamma$,
		\item[\textnormal{(b)}] $\hat{\bm{\beta}}^R(\Lambda, \bm{\beta}_0) \toprob \bm{\beta}$,
		\item[\textnormal{(c)}] $\hat{{\Sigma}}_u^R \toprob \Sigma_u$,
		\item[\textnormal{(d)}] $\sqrt{T} \left( \hat{\bm{\beta}}^R(\Lambda, \bm{\beta}_0) - \bm{\beta} \right) 
		\todist 
		\mathcal{N} \left( \Gamma\inv \Lambda_0 (\underline{\bm{\beta}}_0 - \bm{\beta}), \Gamma\inv \otimes \Sigma_u \right)$.
	\end{itemize}
\end{theorem}

Theorem \ref{th:MAIN_THEOREM} considers the most general case, and, as previously mentioned, gives the asymptotic distribution of $\hat{\bm{\beta}}^R$ under rather weak conditions for the regularizer $\Lambda$. The resulting normal limit distribution is clearly dependent on the unknown model parameters $\bm{\beta}$, complicating inference. 

However, it is possible -- under strengthened assumptions for $\Lambda$ or $\bm{\beta}_0$ -- for $\hat{\bm{\beta}}^R$ to have a zero-mean Gaussian limit distribution.

\begin{theorem}\label{th:Corollary_MAIN_THEOREM}
	In the setting of Theorem \ref{th:MAIN_THEOREM}, results (a)-(c) hold and (d) simplifies to
	\begin{itemize}
		\item[\textnormal{(d$'$)}] $\sqrt{T} \left( \hat{\bm{\beta}}^R(\Lambda, \bm{\beta}_0 \right) - \bm{\beta}) \todist \mathcal{N} \left( 0, \Gamma\inv \otimes \Sigma_u \right)$
	\end{itemize}
	if either
	\begin{itemize}
		\item[\textnormal{(i)}] ${\Lambda} = o_P\left( T^{-1/2} \right)$,
		\item[\textnormal{(ii)}] ${\Lambda} = O_P\left( T^{-1/2} \right)$ and $ \bm{\beta}_0 - \bm{\beta} = o_p(1)$.
	\end{itemize}
\end{theorem}

The following corollary is immediate.

\begin{corollary}
	Let $\hat{\bm{\beta}}_0$ be a consistent and asymptotically normal estimator of $\bm{\beta}$. Then, under condition \textnormal{(i)} or \textnormal{(ii)} of Theorem \ref{th:Corollary_MAIN_THEOREM} results \textnormal{(a)}-\textnormal{(d$'$)} hold.
\end{corollary}

\subsection{Joint Inference}

To handle smooth transformations of VAR coefficients, such as impulse responses \citep{lutkepohlAsymptoticDistributionsImpulse1990}, I also derive a standard joint limit result for both $\hat{\bm{\beta}}^R$ and the variance estimator $\hat{\Sigma}_u^R$.

\begin{theorem}\label{Joint_MAIN_THEOREM}
	Let $\hat{\bm{\sigma}}^R = \vect(\hat{\Sigma}_u^R)$ and $\bm{\sigma} = \vect(\Sigma_u)$. Under the assumptions of Theorem \ref{th:MAIN_THEOREM},
	\begin{equation*}
		\sqrt{T} \begin{bmatrix}
			\hat{\bm{\beta}}^R - \bm{\beta} \\
			\hat{\bm{\sigma}}^R - \bm{\sigma}
		\end{bmatrix}
		\todist 
		\mathcal{N}\left(
		\begin{bmatrix}
			\Gamma\inv \Lambda_0 (\underline{\bm{\beta}}_0 - \bm{\beta}) \\
			0
		\end{bmatrix},
		\begin{bmatrix}
			\Gamma\inv \otimes \Sigma_u & 0 \: \\
			0 & \Omega \:
		\end{bmatrix}
		\right).
	\end{equation*}
	Under assumption (1) or (2) of Theorem \ref{th:Corollary_MAIN_THEOREM},
	\begin{equation*}
		\sqrt{T} \begin{bmatrix}
			\hat{\bm{\beta}}^R - \bm{\beta} \\
			\hat{\bm{\sigma}}^R - \bm{\sigma}
		\end{bmatrix}
		\todist 
		\mathcal{N}\left(
		0,
		\begin{bmatrix}
			\Gamma\inv \otimes \Sigma_u & 0 \: \\
			0 & \Omega \:
		\end{bmatrix}
		\right),
	\end{equation*}
	where $\Omega = \E\big[\vect(u_t u_t')\, \vect(u_t u_t')'\big] - \bm{\sigma} \bm{\sigma}'$.
\end{theorem}

This result is key as it allows, under the stated assumptions on the penalizer, to construct valid asymptotic confidence intervals and, specifically, perform impulse response inference, as done in the simulations of Section \ref{section:monte_carlo_simulations} using the Delta Method, c.f. \cite{lutkepohl2005new}.

\subsection{Cross-validation}

In practice, the choice of ridge penalty is often data-driven, and cross-validation is an very popular approach to select $\Lambda$. I now turn to the properties of CV as applied to $\hat{\bm{\beta}}^R(\Lambda)$.

For simplicity, assume that $y_t$ is an AR($p$) process, that is, $K = 1$. In this setting,
\begin{equation*}
	\hat{\bm{\beta}}^R(\Lambda) 
	= 
	\left( \frac{Z Z'}{T} + \Lambda \right)\inv \frac{Z\bof{y}}{T} ,
\end{equation*}
where $\Lambda = \diag\{\lambda_1, \ldots, \lambda_p\}$. Following \cite{Patil2021}, the \textit{prediction error} of ridge estimator $\hat{\bm{\beta}}^R(\Lambda)$ given penalty $\Lambda$ is 
\begin{equation*}\label{eq:pred_err}
	\textnormal{Err}\left( \hat{\bm{\beta}}^R(\Lambda) \right) 
	:= 
	\E_{\,\tilde{y}, \tilde{z}}\left[ \left(\widetilde{y} - \widetilde{z}' \hat{\bm{\beta}}^R(\Lambda) \right)^2 \big\vert Z, \bof{y} \right] ,
\end{equation*}
where $\widetilde{y}$ and $\widetilde{z}$ are random variables from an independent copy of $y_t$. In particular, $\widetilde{z}$ is the vector of $p$ lags of $\widetilde{y}$. 
Moreover, the error curve for $\Lambda$ is given by
\begin{equation*}
	\textnormal{err}(\Lambda) := \textnormal{Err}\left( \hat{\bm{\beta}}^R(\Lambda) \right) .
\end{equation*}
The prediction error is crucial because it allows to determine the oracle optimal penalization,
\begin{equation*}
	\Lambda^\ast := \arg\min_{\Lambda \succeq 0}\ \textnormal{err}(\Lambda) .
\end{equation*}
Clearly, $\textnormal{err}(\Lambda)$ is unavailable in practice and $\Lambda^\ast$ must be substituted with a feasible alternative. Cross-validation proposes to construct a collection of paired, non-overlapping subsets of the sample data such that the first subset of the pair (estimation set) is used to estimate the model, while the second (validation set) is used to provide an empirical estimate of the prediction error. The CV penalty is then selected to minimize the total error over validation sets. A very popular approach to build cross-validation subsets is $k$-fold CV, wherein the sample is split into $k$ blocks, so-called \textit{folds}, of sequential observations (possibly after shuffling the data). Each fold determines a validation set, and is paired with its complement, which gives the estimation set. For more details, see e.g. \cite{hastie2009elements}.

Again with the intent of keeping complexity low -- as this work is not focused on cross-validation -- I will make the additional simplifying assumption that CV is implemented with two folds and one pair.  Specifically, the first fold is the estimation set, where $Z$ and $\bof{y}$ are constructed and $\hat{\bm{\beta}}^R(\Lambda)$ is estimated. The second fold is the validation set and yields $\widetilde{Z}$, $\widetilde{\bof{y}}$, where $\widetilde{Z} \in \Real^{p \times \widetilde{T}}$ and $\widetilde{\bof{y}} \in \Real^{\widetilde{T}}$. To account for dependence, a buffer of $m$ observations between validation and estimation folds is introduced. The last observation of $y_t$ in the estimation set is $y_T$, while the first observation in validation set is $\widetilde{y}_1 := y_{T+m+1}$, that is, the total number of available observations is $T + m + \widetilde{T} + 2p + 1$. This is a stylized version of the CV setup of \cite{Burman1994} -- also called \textit{$m$-block} or \textit{non-dependent cross-validation} in \cite{bergmeirNoteValidityCrossvalidation2018} -- and is effectively equivalent to an out-of-sample (OOS) validation scheme. 
%
%
Thus, the 2-fold $m$-buffered CV error curve is
\begin{equation}\label{eq:cv_pred_err}
	\textnormal{cv2}_m(\Lambda) 
	:= 
	\frac{1}{\widetilde{T}} \sum_{s = 1}^{\widetilde{T}} \left( \widetilde{y}_s - \widetilde{z}_s' \hat{\bm{\beta}}^R(\Lambda) \right)^2 .
\end{equation}

\begin{theorem}\label{th:cv2_m_asconv_err}
	Under Assumptions A-C, for every $\Lambda$ in the cone of diagonal positive definite penalty matrices with diagonal entries in $(\lambda_{\min}, \infty)$, $\lambda_{\min} \geq 0$, it holds that 	
	\begin{equation*}
		\textnormal{cv2}_m(\Lambda) - \textnormal{err}(\Lambda) \overset{a.s.}{\longrightarrow} 0
	\end{equation*}
	as $T, \widetilde{T} \to \infty$. Furthermore, the convergence is uniform in $\Lambda$ over compact subsets of penalty matrices.
\end{theorem}

In the current setup, the joint limit $T, \widetilde{T} \to \infty$ should be thought as $\widetilde{T} / T \to \gamma \in (0, 1)$, where aspect ratio $\gamma$ determines the balance of the cross-validation split. 

\begin{remark}
	Under Assumption~C, $\omega_{K}(\hat{\Gamma}) > 0$ for $T$ large. Therefore, the bounds derived in the proof of Theorem~\ref{th:cv2_m_asconv_err} are finite even if $\Lambda = 0$. In fact, it is easily seen that the behavior of $\textnormal{err}(\Lambda)$ and $\textnormal{cv2}_m(\Lambda)$ is consistent at the endpoints $\Lambda = 0$ and $\Lambda \to \infty$, c.f. \cite{Patil2021}. Observe that
	\begin{equation*}
		\textnormal{cv2}_m(\Lambda) \to \Sigma_{u}
		\quad\textnormal{and}\quad
		\textnormal{err}(\Lambda) \to \Sigma_{u}
	\end{equation*}
	as $\Lambda \to 0$, while
	\begin{equation*}
		\textnormal{cv2}_m(\Lambda) \to \Gamma
		\quad\textnormal{and}\quad
		\textnormal{err}(\Lambda) \to \Gamma
	\end{equation*}
	as $\Lambda \to \infty$, as needed.
\end{remark}

Theorem~\ref{th:cv2_m_asconv_err} thus shows that $\textnormal{cv2}_m(\Lambda)$ gives an asymptotically valid way to evaluate the prediction error curve, and thus tune $\Lambda$, over any compact set of diagonal positive semi-definite penalization matrices.
Moreover, in Theorem~\ref{th:cv_dependence}, Appendix~\ref{section:cv_under_dependence}, I show that the impact of dependence due to the VAR data generating process is exponentially small for $m$ sufficiently large. This property of $\textnormal{cv2}_m(\Lambda)$ is desirable because it lets one choose $m$ small also in applications with moderate sample sizes, and theoretical justifies the prescription of \cite{Burman1994}.

\subsection{Asymptotically Valid CV}

So far, I have shown that a simple 2-fold CV -- or, equivalently, out-of-sample validation -- correctly estimates the predictive error of the ridge estimator, even under dependence. I turn now to the question of selecting an \textit{asymptotically valid} penalty, that is, a $\Lambda$ such that condition (1) of Theorem \ref{th:Corollary_MAIN_THEOREM} is fulfilled. This enables inference, since one is in a setting where the bias is asymptotically negligible.

The idea is to scale the ridge penalty used at the estimation step of CV by a factor $\sqrt{T}$, so that the validated penalty converges to zero at an appropriate rate as both $T$ and $\widetilde{T}$ grow. In other words, an over-smoothed ridge regression turns out to be key when studying cross-validation. To derive this result, first let
\begin{equation*}
	\hat{\bm{\beta}}^R_\blacklozenge(\Lambda) 
	:= 
	\left( \frac{Z Z'}{T} + \sqrt{T} \Lambda \right)\inv \frac{Z\bof{y}}{T}
\end{equation*}
be the \textit{over-smoothed ridge estimator}.

\begin{theorem}\label{th:asym_valid_cv_penalty}
	Under Assumptions A-C, let $\mathcal{I}_\lambda$ be the compact set of diagonal positive semidefinite penalization matrices $\Lambda$ such that $\norm{\Lambda}_{\max} \leq \lambda < \infty$. It holds
	\begin{equation*}
		\Lambda_\blacklozenge^*
		:=
		\arg\min_{\Lambda \in \mathcal{I}_\lambda}\, \textnormal{Err}\left( \hat{\bm{\beta}}^R_\blacklozenge(\Lambda) \right) 
		= 
		o_p(T^{-1/2}) .
	\end{equation*}
\end{theorem} 

\begin{remark}
	The previous theorem is stated in terms of the oracle predictive error $\textnormal{Err}\left( \hat{\bm{\beta}}^R_\blacklozenge \left( \widetilde{\Lambda} \right) \right)$, which equals the 2-fold CV error curve up to a factor of order $O_P({\widetilde{T}}^{-1/2})$. Therefore, assuming that the CV aspect ratio $\gamma$ is strictly between zero and one, the result of Theorem~\ref{th:asym_valid_cv_penalty} also directly generalizes to an empirically cross-validated penalty.
\end{remark}

\section{Inference with Shrinkage}\label{section:inference_shrinkage}

\cite{fuAsymptoticsLassotypeEstimators2000} have argued that results such as Theorems \ref{th:MAIN_THEOREM} and \ref{th:Corollary_MAIN_THEOREM} portray penalized estimators in a somewhat unfair light, because they result in asymptotic distributions showing no bias-variance trade-off. Indeed, they show that ridge shrinkage yields estimates with asymptotic variance no different than that of least squares. Of course, in finite samples shrinkage has an effect on $\Gamma\inv \otimes \Sigma_u$ since $\hat{{\Sigma}}^R_T$ is used in place of $\hat{\Sigma}^{LS}_T$ to estimate the error term variance matrix. To better understand the value of ridge penalization in practice, one should therefore consider the situation where a number of VAR coefficients are small, but not necessarily zero. 

Formally, assume that for some $0 < n \leq p$ one can partition the VAR coefficients as $\bm{\beta} = (\bm{\beta}_1', \bm{\beta}_2')'$, where $\bm{\beta}_1 \in \Real^{K^2(p-n)}$ and $\bm{\beta}_2 \in \Real^{K^2 n}$, and assume that $\bm{\beta}_2 = T^{-(1/2 + \delta)}\,\bm{b}_2$ for $\delta > 0$ and $\bm{b}_2 \in \Real^{K^2 n}$ is fixed. Such ordered partitioning of $\bm{\beta}$ is without loss of generality.\footnote{The dimensions of $\bm{\beta}_1$ and $\bm{\beta}_2$ are chosen to be multiples of $K^2$ to better conform to the lag-adapted setting. This choice is also without loss of generality and simplifies exposition.} 
In this context, it is desirable to penalize $\bm{\beta}_1$ and $\bm{\beta}_2$ differently when constructing the ridge penalty. 
From a practical perspective, one can consider, for example, the case of a VAR($\infty$) model derived by inverting a stable VARMA($p, q$) process: for $i$ sufficiently large, coefficient matrices $A_i$ decay exponentially to zero.\footnote{This result follows from a straightforward generalization of Lemma \ref{lemma:exp_decay_MA} in Appendix \ref{appendix:cross_validation}. The choice of norm to measure such decay is not fundamental as they are equivalent given that dimension $K$ is fixed.} In a finite sample, an asymptotic framework with non-negligible penalization of higher-order lag coefficients can be more appropriate than that of Theorem \ref{th:MAIN_THEOREM}. This approach to inference is also in the vein of \cite{demolForecastingUsingLarge2008a}, who argue for explicit lag penalization into BVAR priors on similar grounds.
Thus, the idea of partitioned penalization exploits prior information on the structure of autoregressive coefficients to asymptotically improve on the bias-variance trade-off. In the context of maximum-likelihood estimation, the use of appropriate and plausible model restrictions to improve efficiency by shrinkage, rather then perform hypothesis testing, has also been discussed by \cite{hansenEfficientShrinkageParametric2016}.

Let $\Lambda = \text{diag}\{(L_1', L_2')'\} \otimes I_K$ where $L_1\in \Real^{K^2 (p-n)}_+$ and $L_2 \in \Real^{K^2 n}_+$. Assume that 
\begin{equation*}
	L_1 = o_P(T^{-1/2})
	\quad\textnormal{and}\quad
	L_2 \toprob \overline{L}_2
\end{equation*}
for a fixed vector $\overline{L}_2 \in \Real^{K^2 n}_+$. In particular, letting $\Lambda_1 = \text{diag}\{L_1\}$ and $\Lambda_2 = \text{diag}\{L_2\}$,
\begin{equation}\label{eq:Lambda_1:2}
	{\Lambda} = \begin{bmatrix}
		\Lambda_1 & 0 \\
		0 & \Lambda_2
	\end{bmatrix} \otimes I_K
	\toprob
	\overline{\Lambda} \otimes I_K
	\quad
	\text{where}
	\quad
	\overline{\Lambda}
	=
	\begin{bmatrix}
		0 & 0 \\
		0 & \overline{\Lambda}_2
	\end{bmatrix}, 
	\;
	\overline{\Lambda}_2 
	\succ 0.
\end{equation}
One can now develop an asymptotic result which shows non-negligible shrinkage in the limit distribution of the ridge estimator. For simplicity of exposition, here I will assume that ridge centering $\bm{\beta}_0$ is chosen to be zero.
\begin{theorem}\label{th:MAIN_TEOREM_with_shrink}
	In the setting of Theorem \ref{th:MAIN_THEOREM}, assume that, for $0 < n \leq p$,
	\begin{itemize}
		\item[(i)] $\bm{\beta} = (\bm{\beta}_1', \bm{\beta}_2')'$ where $\bm{\beta}_1 \in \Real^{K^2(p-n)}$ and $\bm{\beta}_2 = T^{-(1/2 + \delta)}\,\bm{b}_2$ for $\delta > 0$, $\bm{b}_2 \in \Real^{K^2 n}$ fixed.
		\item[(ii)] $\Lambda = \mathrm{diag}\{(L_1', L_2')'\}$ where $L_1\in \Real^{K^2 (p-n)}_+$ and $L_2 \in \Real^{K^2 n}_+$.
		\item[(iii)] $L_1 = o_P(T^{-1/2}) $ and $L_2 \toprob \overline{L}_2$ as $T \to \infty$.
		\item[(iv)] $\bm{\beta}_0 = 0$.
	\end{itemize}
	Let $\Gamma_{\overline{\Lambda}} = \Gamma + \overline{\Lambda}$ where $\overline{\Lambda} \succeq 0$ is given by \eqref{eq:Lambda_1:2}. Then, results (a)-(c) hold and 
	\begin{itemize}
		\item[(d$\,''$)] $\sqrt{T} \left( \hat{\bm{\beta}}^R(\Lambda, \bm{\beta}_0) - \bm{\beta} \right) 
		\todist 
		\mathcal{N}\left( 0, \Gamma_{\overline{\Lambda}}\inv\, \Gamma\, \Gamma_{\overline{\Lambda}}\inv \otimes \Sigma_u \right)$
	\end{itemize}
\end{theorem}

It is easy to see that indeed the term $\Gamma_{\overline{\Lambda}}\inv\, \Gamma\, \Gamma_{\overline{\Lambda}}\inv $ in Theorem \ref{th:MAIN_TEOREM_with_shrink} is weakly smaller than $\Gamma\inv$ in the positive-definite sense. Note that
\begin{align*}
	\Gamma_{\overline{\Lambda}}\inv\, \Gamma\, \Gamma_{\overline{\Lambda}}\inv \:\preceq\: \Gamma \inv \: & \iff \: (\Gamma + \overline{\Lambda})\inv \Gamma \:\preceq\: \Gamma\inv (\Gamma + \overline{\Lambda}) \\
	& \iff \: I_{K^2p} - (\Gamma + \overline{\Lambda})\inv \overline{\Lambda} \:\preceq\: I_{K^2p} + \Gamma\inv \overline{\Lambda} \\
	& \iff \: 0 \:\preceq\: ((\Gamma + \overline{\Lambda})\inv + \Gamma\inv)\overline{\Lambda}
\end{align*}
The last inequality is true by definition of $\overline{\Lambda}$. Shrinkage gains are concentrated at the components that have non-zero asymptotic shrinkage, i.e. those penalized by $\overline{L}_2$. 

\begin{remark}
	A key point in the application of Theorem \ref{th:MAIN_TEOREM_with_shrink} is identification of $\bm{\beta}_1$ and $\bm{\beta}_2$. In practice, one may then proceed in two ways. As discussed in Section \ref{section:bayesian_ridge}, one can see the ridge approach as a frequentist ``counterpart'' to implementing a Bayesian prior. Therefore, the researcher may split $\bm{\beta}$ into subsets of small and large parameters based on economic intuition, domain knowledge or preliminary information. Alternatively, in the following section I show that cross-validation is automatically able to tune $\Lambda$ appropriately.
\end{remark}

\subsection{Cross-validation with Partitioned Coefficients}

One can use the same approach applied to derive Theorem~\ref{th:asym_valid_cv_penalty} in order to show that cross-validating the RLS estimator with $\textnormal{Err} ( \hat{\bm{\beta}}^R_\blacklozenge(\Lambda) )$ is also asymptotically valid. 

\begin{corollary}\label{corollary:asym_valid_cv_shrink}
	Consider the setup of Theorem~\ref{th:MAIN_TEOREM_with_shrink} and assume that the assumptions of Theorem~\ref{th:asym_valid_cv_penalty} are met. It holds
	\begin{equation*}
		\begin{bmatrix}
			\Lambda_{1,\blacklozenge} & 0 \\
			0 & \Lambda_{2,\blacklozenge} 
		\end{bmatrix}
		:=
		\arg\min_{\Lambda \in \mathcal{I}_\lambda}\, \textnormal{Err}\left( \hat{\bm{\beta}}^R_\blacklozenge(\Lambda) \right) 
		= 
		\begin{bmatrix}
			o_p(T^{-1/2}) & 0 \\
			0 & o_P(1)
		\end{bmatrix}
	\end{equation*}
	Moreover, any $\Lambda_{2,\blacklozenge}$ such that $0 \preceq \Lambda_{2,\blacklozenge} \preceq \lambda I$ is asymptotically valid.
\end{corollary}

In theory, one would like to be able to quantify the gains obtained in the asymptotic shrinkage setup of Theorem~\ref{th:MAIN_TEOREM_with_shrink} compared to the standard setting of Theorems~\ref{th:MAIN_THEOREM} and \ref{th:Corollary_MAIN_THEOREM}, particularly when using cross-validation. Unfortunately, it is in general hard to study the cross-validation error loss even in setups without dependence. \cite{Stephenson2021} in fact show that the ridge leave-one-out CV loss is not generally convex. This suggests that studying the behavior of CV when penalizing with a diagonal anisotropic $\Lambda$ can be a very complex task in a finite sample setup.

\section{Simulations}\label{section:monte_carlo_simulations}

To study the performance of ridge-regularized estimators, I now perform simulation exercises focused on impulse response functions (IRFs). Throughout the experiments I will consider structural impulse responses, and I assume that identification can be obtained in a recursive way \citep{kilian2017structural}, which is a widely used approach for structural shock identification in macroeconometrics.

I consider two setups:
\begin{enumerate}
	\item The three-variable VARMA(1,1) design of \cite{kilianHowReliableAre2011}, representing a small-scale macro model. I term this setup ``A''.
	\item A VAR(5) model in levels, using the model specification of \cite{giannonePriorSelectionVector2015} with the dataset of \cite{hansenSteinCombinationShrinkage2016a} consisting of $K=7$ variables in levels.\footnote{The dataset is supplied by the author at \url{https://users.ssc.wisc.edu/~bhansen/progs/var.html}. While the data provided by \cite{hansenSteinCombinationShrinkage2016a} includes releases until 2016. I do not include more recent quarterly data since this is a simulation exercise. Moreover, due to the effects of the COVID-19 global pandemic, an extended sample would likely only add data released until Q4 2019 due to overwhelming concerns of a break point.} I term this setup ``B''. For the ease of exposition, in the discussion I will tabulate results only for three variables -- real GDP, investment and federal funds rate -- but complete tables can be found in Appendix \ref{appendix:additional_tables}.
\end{enumerate}
The specification of \cite{kilianHowReliableAre2011} has already been extensively used in the literature as a benchmark to gauge the basic properties of inference methods. On the other hand, the estimation task of \cite{giannonePriorSelectionVector2015} involves more variables and a higher degree of persistence. This setting is useful to evaluate the effects of ridge shrinkage when applied to realistic macroeconomic questions. It is also a suitable test bench to compare Bayesian methods with frequentist ridge.

\paragraph{Estimators.}
For frequentist methods, I include both $\hat{\bm{\beta}}^R$ and $\hat{\bm{\beta}}^{RGLS}$ ridge estimators as well as the local projection estimator of \cite{jordaEstimationInferenceImpulse2005}. For Bayesian methods, I implement both the Minnesota prior approach of \cite{banburaLargeBayesianVector2010} with stationary prior and the hierarchical prior BVAR of \cite{giannonePriorSelectionVector2015}.\footnote{To estimate hierarchical prior BVARs I rely on the original MATLAB implementation provided by \cite{giannonePriorSelectionVector2015} on the authors' website at \url{http://faculty.wcas.northwestern.edu/gep575/GLPreplicationWeb.zip}.}
The full list of method I consider is given in Table \ref{table:methods}.
To make methods comparable, I have extended the ridge estimators to include an intercept in the regression. A precise discussion regarding the tuning of penalties and hyperparameters of all methods can be found in Supplementary Appendix \ref{section:simulation_details}.

\begin{table}[t]
	\small
	\centering
	\renewcommand{\arraystretch}{1.4}
	\linespread{0.8}\selectfont\centering
	\caption{List of Estimation Methods}
	\begin{tabular}{l c p{0.65\linewidth}}
		\toprule
		Type & Name & Description \\
		\midrule
        Frequentist & LS & Least squares estimator \\
		  & RIDGE & Ridge estimator, CV penalty \\
		& RIDGE-GLS & GLS ridge estimator, CV penalty \\
		& RIDGE-AS & Ridge estimator with asymptotic shrinkage, CV penalty \\
		& LP & Local projections with Newey-West covariance estimate \\
		\midrule
		Bayesian & BVAR-CV & Litterman-Minnesota Bayesian VAR, CV tightness prior \\
		& H-BVAR & Hierarchical Bayesian VAR of \cite{giannonePriorSelectionVector2015}\\
		\bottomrule
	\end{tabular}
	\renewcommand{\arraystretch}{1}
	\label{table:methods}
\end{table}

\subsection{Pointwise MSE}

The first two simulation designs explore the MSE performance of ridge-type estimators versus alternatives. Let $\theta_{km}(h)$ be the horizon $h$ structural IRF for variable $k$ given a unit shock from variable $m$. To compute the MSE for each $k$, define
\begin{equation*}
	\textnormal{MSE}_k(h) := \sum_{m=1}^K \E\left[ \left( \hat{\theta}_{km}(h) - \theta_{km}(h)\right)^2 \right] ,
\end{equation*}
which is the total MSE for variable $k$ over all possible structural shocks. In simulations, I use $B$ replications to estimate the expectation. All MSEs are normalized by the mean squared error of the least squares estimator.

\paragraph{Setup A.}
A time series of length $T = 200$ is generated a number $B = 10\,000$ of times for replication. All VAR estimators are computed using $p=10$ lags, while LPs include $q = 10$ regression lags. Table \ref{table:tab:ir_rootmse_1} shows relative MSEs for this design. It is important to notice that in this situation GLS ridge has remarkably low performance at horizon $h=1$ compared to other methods. The primary issue is that $\Sigma_u$ features strong correlation between components, and thus the diagonal lag-adapted structure does not shrink along the appropriate directions. This is much less prominent as the horizon increases due to the fact that impulse responses eventually decay to zero, since the underlying VARMA DGP is stationary. While there is no clear ranking, the MSE of the baseline ridge VAR estimator is in between those of the BVAR and hierarchical BVAR approaches. The degrading quality of local projection estimates are mainly due to the smaller samples available in regressions at each increasing horizon \citep{kilianHowReliableAre2011}. This behavior is one of the prime reasons behind the development of LP shrinkage estimators, like that proposed in \cite{plagborg2016essays} or the SLP estimator of \cite{barnichonImpulseResponseEstimation2019}.

\begin{table}[t!]
	\centering
	\renewcommand{\arraystretch}{1.05}
    \caption{MSE Relative to OLS -- Setup A}
	\begin{tabular}{ccccccccc}
		\toprule
		Variable & Method & $h$ = 1 & $h$ = 4 & $h$ = 8 & $h$ = 12 & $h$ = 16 & $h$ = 20 & $h$ = 24 \\
		\midrule
		& RIDGE & 0.97 & 0.74 & 0.64 & 0.64 & 0.65 & 0.63 & 0.60 \\
		Investment & RIDGE-GLS & 5.16 & 0.89 & 0.55 & 0.47 & 0.44 & 0.41 & 0.38 \\
		Growth & LP & 1.00 & 1.05 & 1.13 & 1.52 & 2.15 & 3.20 & 4.87 \\
		& BVAR-CV & 1.55 & 0.84 & 0.70 & 0.70 & 0.71 & 0.70 & 0.66 \\
		& H-BVAR & 1.80 & 0.66 & 0.53 & 0.52 & 0.54 & 0.53 & 0.50 \\
		\midrule
		& RIDGE & 0.93 & 0.78 & 0.69 & 0.68 & 0.67 & 0.64 & 0.59 \\
		Deflator & RIDGE-GLS & 2.43 & 0.83 & 0.59 & 0.52 & 0.48 & 0.44 & 0.40 \\
		& LP & 1.00 & 1.05 & 1.13 & 1.44 & 1.99 & 2.90 & 4.47 \\
		& BVAR-CV & 1.03 & 0.89 & 0.74 & 0.73 & 0.73 & 0.70 & 0.66 \\
		& H-BVAR & 1.01 & 0.70 & 0.58 & 0.56 & 0.55 & 0.53 & 0.50 \\
		\midrule
		& RIDGE & 0.94 & 0.76 & 0.66 & 0.66 & 0.66 & 0.64 & 0.60 \\
		Paper Rate & RIDGE-GLS & 1.80 & 0.87 & 0.59 & 0.52 & 0.47 & 0.43 & 0.39 \\
		& LP & 1.00 & 1.05 & 1.13 & 1.46 & 1.99 & 2.86 & 4.31 \\
		& BVAR-CV & 0.87 & 0.87 & 0.74 & 0.73 & 0.73 & 0.71 & 0.66 \\
		& H-BVAR & 0.81 & 0.69 & 0.57 & 0.55 & 0.56 & 0.54 & 0.51 \\
		\bottomrule
	\end{tabular}
	\renewcommand{\arraystretch}{1}
	\label{table:tab:ir_rootmse_1}
    ~\\[1.5em]
\end{table}

\begin{table}[h!]
	\centering
	\renewcommand{\arraystretch}{1.05}
	\caption{MSE Relative to OLS -- Setup B}
	\begin{tabular}{ccccccccc}
		\toprule
		Variable & Method & $h$ = 1 & $h$ = 4 & $h$ = 8 & $h$ = 12 & $h$ = 16 & $h$ = 20 & $h$ = 24 \\
		\midrule
		& RIDGE & 1.11 & 1.08 & 1.16 & 1.06 & 0.90 & 0.89 & 0.94 \\
		& RIDGE-GLS & 1.16 & 1.00 & 0.99 & 1.00 & 0.93 & 0.93 & 0.95 \\
		Real GDP & LP & 1.00 & 1.14 & 1.37 & 1.52 & 1.72 & 1.98 & 2.24 \\
		& BVAR-CV & 0.90 & 0.87 & 1.04 & 1.01 & 0.92 & 0.92 & 0.98 \\
		& H-BVAR & 0.83 & 0.62 & 0.78 & 0.73 & 0.62 & 0.62 & 0.68 \\
		\midrule
		& RIDGE & 1.49 & 1.27 & 1.17 & 0.99 & 0.70 & 0.73 & 1.61 \\
		& RIDGE-GLS & 1.34 & 1.14 & 1.02 & 1.02 & 0.86 & 0.82 & 0.86 \\
		Investment & LP & 1.00 & 1.15 & 1.40 & 1.63 & 2.03 & 2.76 & 3.59 \\
		& BVAR-CV & 1.51 & 1.01 & 0.97 & 0.97 & 0.93 & 1.08 & 1.24 \\
		& H-BVAR & 1.06 & 0.68 & 0.69 & 0.66 & 0.63 & 0.87 & 1.14 \\
		\midrule
		& RIDGE & 2.17 & 1.21 & 0.96 & 0.93 & 1.03 & 4.00 & 53.18 \\
		& RIDGE-GLS & 1.21 & 1.04 & 0.90 & 0.93 & 0.90 & 0.88 & 0.91 \\
		Fed Funds & LP & 1.00 & 1.18 & 1.51 & 1.71 & 1.97 & 2.44 & 2.99 \\ 
		Rate & BVAR-CV & 0.92 & 0.94 & 0.91 & 0.90 & 0.86 & 0.87 & 0.92 \\
		& H-BVAR & 0.75 & 0.77 & 1.32 & 1.38 & 1.25 & 1.15 & 1.20 \\
		\bottomrule
	\end{tabular}
	\renewcommand{\arraystretch}{1}
	\label{table:tab:ir_rootmse_2}
\end{table}

\paragraph{Setup B.} 
Using the data of \cite{hansenSteinCombinationShrinkage2016a}, I estimate and simulate a stationary but highly persistent VAR(5) model using the same sample size and number of replications of Setup A. For all methods, $p = 5$ lags are used, so that VAR estimators are correctly specified. In this setup, unlike in the previous experiment, one can clearly notice that impulse responses computed via cross-validated ridge show increasing MSE as horizon $h$ grows. There are two main reasons behind this behavior. First, the chosen setup features a very persistent data generating process, as the largest root of the underlying VAR model is $0.9945$. This means that the true IRFs revert to zero only over long horizons, while lag-adapted ridge estimates yields models with lower persistence and thus flatter impulse responses. Secondly, the dataset from \cite{hansenSteinCombinationShrinkage2016a} is not normalized, and the included series have markedly heterogenous variances. Since GLS ridge shrinks along covariance-rotated data, shrinkage is adjusted according to each series variance, unlike that baseline ridge estimator $\hat{\bm{\beta}}^R$. The MSE for the Fed Fund Rate impulse responses shows that the pointwise difference between baseline and GLS ridge can be severe for long horizon IRFs when the DGP is highly persistent. On short horizons, Bayesian estimators perform on par or better than baseline least squares estimates, while at longer horizons differences are less stark. It is, however, clear that the hierarchical prior BVAR of \cite{giannonePriorSelectionVector2015} shows the overall best results. As in the previous setup, local projections show degrading performance at larger horizons.

\subsection{Confidence Intervals}

I now try and evaluate whether ridge shrinkage has a negative impact on inference. There have also been recent contributions directly aimed at studying shrinkage effects. \cite{liLocalProjectionsVs2022} give an extensive treatment of the issue in terms of bias-variance trade-off, showing that, to a large extent, shrinkage is desirable unless bias is a primary and sensitive concern.
Using the same simulation setups as in the previous section, I investigate coverage and size properties of pointwise CIs constructed using the methods in Table \ref{table:methods}. All confidence intervals are constructed with nominal 90\% level coverage. 

In this set of simulations, I swap GLS ridge for the asymptotic shrinkage ridge estimator, $\hat{\bm{\beta}}^{R}_{as}$, see Section \ref{section:inference_shrinkage}, since the latter allows for a partially non-negligible penalization in the limit. To implement $\hat{\bm{\beta}}^{R}_{as}$, one needs to choose a partition of $\bm{\beta}$ which identifies asymptotically negligible coefficient. To do this, I split $\bm{\beta}$ by lag and penalize all coefficients with lag orders greater than a given threshold $\overline{p}$, such that $1 < \overline{p} < p$. In setup A, I choose $\overline{p} = 6$, while in setup B I set $\overline{p} = 3$.
In Bayesian methods, including the cross-validated Minnesota BVAR, I construct high-probability intervals by drawing from the posterior. Comparison between frequentist CIs and Bayesian posterior densities is not generally valid, because they are not analogous concepts. Therefore, the discussion below is intended to highlight differences in \textit{structure} between ridge approaches.

\begin{table}[t!]
	\centering
	\renewcommand{\arraystretch}{0.9}
	\caption{Impulse Response Inference -- Setup A -- CI Coverage}
	\begin{tabular}{ccccccccc}
		\toprule
		Variable & Method & $h$ = 1 & $h$ = 4 & $h$ = 8 & $h$ = 12 & $h$ = 16 & $h$ = 20 & $h$ = 24 \\
		\midrule
		& LS & 0.88 & 0.88 & 0.87 & 0.88 & 0.91 & 0.93 & 0.94 \\
		& RIDGE & 0.90 & 0.92 & 0.94 & 0.93 & 0.94 & 0.95 & 0.95 \\
		Investment & RIDGE-AS & 0.90 & 0.92 & 0.88 & 0.88 & 0.88 & 0.89 & 0.89 \\
		Growth & LP & 0.88 & 0.97 & 0.99 & 0.99 & 0.99 & 0.99 & 0.99 \\
		& BVAR-CV & 0.77 & 0.88 & 0.88 & 0.90 & 0.92 & 0.94 & 0.96 \\
		& H-BVAR & 0.72 & 0.89 & 0.89 & 0.92 & 0.93 & 0.95 & 0.96 \\
		\midrule
		& LS & 0.88 & 0.87 & 0.86 & 0.88 & 0.91 & 0.92 & 0.94 \\
		& RIDGE & 0.91 & 0.92 & 0.93 & 0.92 & 0.93 & 0.94 & 0.95 \\
		Deflator & RIDGE-AS & 0.91 & 0.91 & 0.88 & 0.88 & 0.87 & 0.87 & 0.88 \\
		& LP & 0.88 & 0.97 & 0.99 & 0.99 & 0.99 & 0.99 & 1.00 \\
		& BVAR-CV & 0.80 & 0.86 & 0.88 & 0.91 & 0.93 & 0.94 & 0.96 \\
		& H-BVAR & 0.84 & 0.88 & 0.90 & 0.92 & 0.94 & 0.95 & 0.97 \\
		\midrule
		& LS & 0.87 & 0.86 & 0.86 & 0.88 & 0.90 & 0.92 & 0.94 \\
		& RIDGE & 0.90 & 0.91 & 0.93 & 0.93 & 0.93 & 0.94 & 0.95 \\
		Paper Rate & RIDGE-AS & 0.89 & 0.90 & 0.89 & 0.88 & 0.88 & 0.88 & 0.88 \\
		& LP & 0.87 & 0.97 & 0.99 & 0.99 & 0.99 & 0.99 & 0.99 \\
		& BVAR-CV & 0.82 & 0.84 & 0.87 & 0.90 & 0.92 & 0.93 & 0.95 \\
		& H-BVAR & 0.88 & 0.88 & 0.90 & 0.92 & 0.93 & 0.95 & 0.96 \\
		\bottomrule
	\end{tabular}\\
	\renewcommand{\arraystretch}{1}
	\label{table:tab:ir_cover_1}
    ~\\[1.5em]
\end{table}

\begin{table}[h!]
	\centering
	\renewcommand{\arraystretch}{0.9}
	\caption{Impulse Response Inference -- Setup A -- CI Length}
    \begin{tabular}{ccccccccc}
		\toprule
		Variable & Method & $h$ = 1 & $h$ = 4 & $h$ = 8 & $h$ = 12 & $h$ = 16 & $h$ = 20 & $h$ = 24 \\
		\midrule
		& LS & 2.99 & 5.11 & 5.78 & 5.35 & 4.79 & 4.17 & 3.56 \\
		& RIDGE & 3.13 & 5.20 & 5.82 & 5.17 & 4.48 & 3.78 & 3.09 \\
		Investment & RIDGE-AS & 3.11 & 5.15 & 4.84 & 4.33 & 3.70 & 3.06 & 2.48 \\
		Growth & LP & 2.99 & 7.50 & 10.97 & 12.89 & 13.99 & 14.55 & 14.70 \\
		& BVAR-CV & 2.84 & 4.48 & 4.70 & 4.38 & 3.99 & 3.56 & 3.11 \\
		& H-BVAR & 2.71 & 4.20 & 4.50 & 4.29 & 3.96 & 3.56 & 3.13 \\
		\midrule
		& LS & 1.19 & 1.92 & 2.23 & 2.14 & 1.94 & 1.71 & 1.46 \\
		& RIDGE & 1.24 & 1.97 & 2.25 & 2.09 & 1.84 & 1.54 & 1.26 \\
		Deflator & RIDGE-AS & 1.24 & 1.95 & 1.95 & 1.78 & 1.52 & 1.25 & 1.01 \\
		& LP & 1.19 & 3.03 & 4.56 & 5.42 & 5.90 & 6.14 & 6.21 \\
		& BVAR-CV & 1.03 & 1.69 & 1.87 & 1.80 & 1.67 & 1.50 & 1.31 \\
		& H-BVAR & 1.01 & 1.64 & 1.83 & 1.79 & 1.67 & 1.51 & 1.33 \\
		\midrule
		& LS & 0.97 & 1.42 & 1.64 & 1.57 & 1.44 & 1.27 & 1.09 \\
		& RIDGE & 1.01 & 1.44 & 1.65 & 1.53 & 1.36 & 1.16 & 0.95 \\
		Paper Rate & RIDGE-AS & 1.01 & 1.43 & 1.42 & 1.31 & 1.13 & 0.94 & 0.77 \\
		& LP & 0.97 & 2.19 & 3.28 & 3.90 & 4.26 & 4.43 & 4.48 \\
		& BVAR-CV & 0.84 & 1.22 & 1.35 & 1.30 & 1.21 & 1.09 & 0.96 \\
		& H-BVAR & 0.85 & 1.21 & 1.34 & 1.31 & 1.22 & 1.10 & 0.97 \\
		\bottomrule
	\end{tabular}
	\renewcommand{\arraystretch}{1}
	\label{table:tab:ir_length_1}
\end{table}

\paragraph{Setup A.}

Simulations with the DGP of \cite{kilianHowReliableAre2011}, presented in Tables \ref{table:tab:ir_cover_1} and \ref{table:tab:ir_length_1}, highlight some of the advantages of applying ridge when performing inference. Focusing on estimator $\hat{\bm{\beta}}^R$, it is clear that CI coverage is in fact higher than the intervals obtained by least squares estimation in all situations. At impact, ridge CIs are larger than the LS baseline, but they shrink as horizons increase. Thus, is IRFs revert relatively quickly to zero, ridge can effectively reduce length while preserving coverage. As discussed in Section \ref{section:ridge_shrinkage}, these gains are inherently local to the DGP -- shrinkage to zero at deep lags embodies correct prior knowledge of a weakly persistent process. For Bayesian estimators, one can note that quantile intervals at small horizons tend to be shorter compared to least squares and ridge methods.

\paragraph{Setup B.}

The effects of ridge shrinkage on a DGP with high persistence are much more severe, as shown in Tables \ref{table:tab:ir_cover_2} and \ref{table:tab:ir_length_2}. Focusing on frequentist ridge, one can observe that close to impact ($h = 1$) ridge has similar or even higher coverage than other methods for real GDP\footnote{This also is the case with consumption and compensation, c.f. Tables \ref{table:tab:ir_cover_2_ext} and \ref{table:tab:ir_length_2_ext} in Appendix \ref{appendix:additional_tables}.} 
However, as the IRF horizon grows, shrinkage often leads to severe undercoverage, with asymptotic shrinkage estimator $\hat{\bm{\beta}}^{R}_{as}$ giving the worst results. In comparison, Bayesian methods are much more reliable at all horizons, although the only estimator that can consistently improve upon the benchmark least squares VAR CIs is the hierarchical prior BVAR of \cite{giannonePriorSelectionVector2015}. The reason behind this is simple enough: the implementation of the Minnesota-prior BVAR I have used here has a white noise prior on all variables, which in this case is far from the truth. Indeed, \cite{banburaLargeBayesianVector2010} implement the same BVAR by tuning the prior to a random walk for very persistent variables in their applications. In this sense the cross-validated BVAR considered -- which is assumed centered at zero -- is really the flip-side of ridge estimators. Therefore, the addition of a prior on the mean of the autoregressive parameters as done by \cite{giannonePriorSelectionVector2015} is a key element to perform shrinkage in high persistence setups in a way that does not systematically undermine asymptotic inference on impulse responses. 

\begin{table}[t!]
	\centering
	\renewcommand{\arraystretch}{0.9}
	\caption{Impulse Response Inference -- Setup B: CI Coverage}
	\vspace{-3pt}
	\begin{tabular}{ccccccccc}
		\toprule
		Variable & Method & $h$ = 1 & $h$ = 4 & $h$ = 8 & $h$ = 12 & $h$ = 16 & $h$ = 20 & $h$ = 24 \\
		\midrule
		& LS & 0.87 & 0.81 & 0.75 & 0.72 & 0.71 & 0.72 & 0.73 \\
		& RIDGE & 0.90 & 0.79 & 0.66 & 0.62 & 0.65 & 0.68 & 0.68 \\
		Real GDP & RIDGE-AS & 0.89 & 0.72 & 0.61 & 0.58 & 0.61 & 0.65 & 0.65 \\
		& LP & 0.87 & 0.93 & 0.94 & 0.94 & 0.93 & 0.93 & 0.91 \\
		& BVAR-CV & 0.70 & 0.71 & 0.63 & 0.64 & 0.71 & 0.75 & 0.76 \\
		& H-BVAR & 0.84 & 0.86 & 0.76 & 0.76 & 0.83 & 0.88 & 0.88 \\
		\midrule
		& LS & 0.87 & 0.82 & 0.76 & 0.73 & 0.75 & 0.82 & 0.87 \\
		& RIDGE & 0.85 & 0.79 & 0.65 & 0.62 & 0.73 & 0.80 & 0.81 \\
		Investment & RIDGE-AS & 0.82 & 0.69 & 0.59 & 0.57 & 0.68 & 0.77 & 0.77 \\
		& LP & 0.87 & 0.94 & 0.94 & 0.95 & 0.94 & 0.94 & 0.94 \\
		& BVAR-CV & 0.70 & 0.73 & 0.67 & 0.71 & 0.77 & 0.81 & 0.83 \\
		& H-BVAR & 0.80 & 0.86 & 0.81 & 0.82 & 0.87 & 0.88 & 0.88 \\
		\midrule
		& LS & 0.85 & 0.83 & 0.80 & 0.78 & 0.77 & 0.79 & 0.80 \\
		& RIDGE & 0.79 & 0.77 & 0.74 & 0.68 & 0.68 & 0.72 & 0.72 \\
		Fed Funds & RIDGE-AS & 0.78 & 0.66 & 0.68 & 0.64 & 0.64 & 0.68 & 0.69 \\
		Rate & LP & 0.85 & 0.94 & 0.96 & 0.96 & 0.95 & 0.94 & 0.93 \\
		& BVAR-CV & 0.76 & 0.72 & 0.76 & 0.77 & 0.77 & 0.81 & 0.83 \\
		& H-BVAR & 0.87 & 0.86 & 0.74 & 0.73 & 0.78 & 0.84 & 0.87 \\
		\bottomrule
	\end{tabular}
	\renewcommand{\arraystretch}{1}
	\label{table:tab:ir_cover_2}
    ~\\[1.5em]
\end{table}

\begin{table}[h!]
	\centering
	\renewcommand{\arraystretch}{0.9}
	\caption{Impulse Response Inference -- Setup B: CI Length (rescaled $\times 100$)}
	\vspace{-3pt}
	\begin{tabular}{ccccccccc}
		\toprule
		Variable & Method & $h$ = 1 & $h$ = 4 & $h$ = 8 & $h$ = 12 & $h$ = 16 & $h$ = 20 & $h$ = 24 \\
		\midrule
		& LS & 0.71 & 1.56 & 2.07 & 2.31 & 2.32 & 2.24 & 2.15 \\
		& RIDGE & 0.79 & 1.56 & 1.85 & 1.95 & 1.92 & 1.85 & 1.77 \\
		Real GDP & RIDGE-AS & 0.74 & 1.31 & 1.65 & 1.76 & 1.75 & 1.70 & 1.64 \\
		& LP & 0.71 & 2.42 & 4.21 & 5.40 & 5.90 & 5.91 & 5.70 \\
		& BVAR-CV & 0.53 & 1.23 & 1.74 & 2.00 & 2.10 & 2.13 & 2.15 \\
		& H-BVAR & 0.58 & 1.36 & 1.87 & 2.16 & 2.32 & 2.44 & 2.55 \\
		\midrule
		& LS & 3.38 & 6.65 & 7.89 & 7.89 & 7.31 & 6.69 & 6.18 \\
		& RIDGE & 3.79 & 6.81 & 6.93 & 6.46 & 5.79 & 5.19 & 4.73 \\
		Investment & RIDGE-AS & 3.59 & 5.57 & 6.11 & 5.77 & 5.21 & 4.72 & 4.34 \\
		& LP & 3.37 & 10.16 & 16.00 & 18.85 & 19.06 & 18.22 & 17.23 \\
		& BVAR-CV & 2.64 & 5.26 & 6.59 & 6.91 & 6.78 & 6.57 & 6.38 \\
		& H-BVAR & 2.89 & 5.74 & 7.08 & 7.54 & 7.63 & 7.60 & 7.58 \\
		\midrule
		& LS & 0.25 & 0.39 & 0.43 & 0.43 & 0.41 & 0.38 & 0.35 \\
		& RIDGE & 0.29 & 0.39 & 0.37 & 0.36 & 0.33 & 0.30 & 0.29 \\
		Fed Funds & RIDGE-AS & 0.27 & 0.31 & 0.33 & 0.32 & 0.30 & 0.28 & 0.27 \\
		Rate & LP & 0.25 & 0.59 & 0.88 & 1.01 & 1.05 & 1.03 & 0.98 \\
		& BVAR-CV & 0.21 & 0.31 & 0.36 & 0.37 & 0.36 & 0.35 & 0.34 \\
		& H-BVAR & 0.23 & 0.36 & 0.42 & 0.44 & 0.45 & 0.45 & 0.46 \\
		\bottomrule
	\end{tabular}
	\renewcommand{\arraystretch}{1}
	\label{table:tab:ir_length_2}
\end{table}

\section{Conclusion}\label{sect:end}

In this paper, I have studied ridge regression and its application to vector autoregressive model estimation in detail. This appears to be the first work that provides a thorough analysis of ridge penalization in the context of time series data, including geometric as well as asymptotic properties. I have also derived results on the validity of cross-validation as a method to select the penalty intensity in practice, and I have shown that CV produces asymptotically valid penalization rates. Finally, I have compared both frequentist and Bayesian ridge formulation in simulations aimed at quantifying the applicability of ridge for impulse response inference.

The key takeaway of this work is that ridge penalization is a useful approach to VAR estimation as long as the chosen penalty structure is well-adapted to the models' structure. Bayesian ridge posteriors are especially flexible, with hierarchical priors also allowing shrinkage towards non-zero coefficient vectors. However, it is important to note that the Bayesian approach also permits the researcher to specify uninformative priors, so that the influence of the priors' hyperparameters is less pronounced. This is not the case in frequentist ridge, c.f. including an explicit non-zero centering vector. However, prior knowledge or a pre-estimation procedure may be available to the researcher, so that ridge can be effectively implemented without the need to implement a BVAR.

To conclude, there are still avenues of research regarding ridge which would be interesting to develop. First and foremost, the high-dimensional setup, for which, however, it seems non-trivial to find a domain of applicability, as discussed in the introduction. Secondly, a more in-depth analysis of cross-validation, especially in the multivariate case, would be extremely valuable. Moreover, both the latter and former topics should be jointly addressed in the context of mild cross-sectional dimension growth, i.e. $K \to \infty$ such that $K / T \to \rho \in (0, 1)$, which is comparable to factor model setups. 


\paragraph{Declaration of Interest Statement.}
The author reports there are no competing interests to declare.


\newpage

{\small
	\onehalfspacing
	\bibliography{./ref_ridge_var_2023}
}

\appendix
\newpage
\setcounter{page}{1}

{\LARGE\hfill Appendix \hfill}

\section{Preliminaries}

\subsection{LS and RLS Estimators.}\label{Supp_Appendix:LS_and_RLS_Estimators}

\cite{lutkepohl2005new}, Chapter 3, shows that the multivariate least squares and GLS estimator of parameter vector $\bm{\beta}$ is given by
\begin{equation*}
	\hat{\bm{\beta}} = ((Z' Z)\inv Z \otimes I_K)\bof{y}
\end{equation*}
as the minimizer of $S(\bm{\beta}) = T\inv \trace[(Y - BZ)' \Sigma_u (Y - BZ)]$. The multivariate ridge-regularized least squares (RLS) -- or, simply, ridge -- estimator considered in this paper is defined to be the minimizer of the regularized problem,
\begin{align*}
	S^R(\bm{\beta}; \Lambda) 
	& = T\inv \trace[(Y - BZ)' (Y - BZ)] + \trace[B' \Lambda B] \\
	& = \frac{\bof{y}'\bof{y}}{T} + \bm{\beta}'\left( \frac{Z Z'}{T} \otimes I_K \right)\bm{\beta} - 2 \bm{\beta}' \frac{(Z \otimes I_K)\bof{y}}{T} + \bm{\beta}' \Lambda \bm{\beta}
\end{align*} 
The first partial derivative,
\begin{equation*}
	\frac{\partial S^R(\bm{\beta}; \Lambda)}{\partial \bm{\beta}} 
	= 
	2 \left( \frac{Z Z'}{T} \otimes I_K \right)\bm{\beta} - 2 \frac{(Z \otimes I_K)\bof{y}}{T} + 2 \Lambda \bm{\beta}
\end{equation*}
gives the normal equations $( T\inv Z Z' \otimes I_K + \Lambda)\bm{\beta} = T\inv (Z \otimes I_K)\bof{y}$. The Hessian $\partial^2 S^R(\bm{\beta})/\partial^2 \bm{\beta} = 2 (T\inv Z Z' \otimes I_K + \Lambda)\bm{\beta}$ is positive definite when $\Lambda > 0$, thus indeed the minimum is achieved by
\begin{equation*}
	\hat{\bm{\beta}}^R(\Lambda) 
	= 
	\left( \frac{Z Z'}{T} \otimes I_K + \Lambda \right)\inv \frac{(Z \otimes I_K)\bof{y}}{T}
\end{equation*}
Identical derivations prove that re-centering the ridge penalty at $\bm{\beta}_0 \in \Real^{K^2p}$ produces the estimator
\begin{equation*}
	\hat{\bm{\beta}}^R(\Lambda, \bm{\beta}_0) 
	= 
	\left( \frac{Z Z'}{T} \otimes I_K + \Lambda \right)\inv \left( \frac{(Z \otimes I_K)\bof{y}}{T} + \Lambda \bm{\beta}_0 \right)
\end{equation*}

\subsection{Structure of the Regularization Matrix}\label{Supp_Appendix:Structure_Lambda}

The vectorized RLS estimator $\hat{\bm{\beta}}^R(\Lambda)$ has maximal flexibility in terms of the regularization structure that matrix $\Lambda = \diag\{\lambda_{1,1}, \ldots, \lambda_{K,p}\}$ $(K^2 p \times K^2 p)$ imposes. Since $\bm{\beta}$ contains all the coefficients of $(A_1, \ldots, A_p)$ it is indeed possible to individually penalize each lag of each series differently. In fact, by relaxing the assumption that $\Lambda$ be a diagonal matrix, even more general penalization structures are possible, although I do not consider them in this paper.

An interesting special case arises if the RLS estimator is instead written in its matrix form\footnote{
	For details in the least squares case, see again \cite{lutkepohl2005new}, Chapter 3. The derivations for the ridge estimator are identical.
},
\begin{equation*}
	\hat{B}^R_{\textnormal{mat}}(\Lambda_{Kp}) = \frac{Y Z'}{T} \left( \frac{Z Z'}{T} + \Lambda_{Kp} \right)\inv	
\end{equation*}
where here it is of note that $\Lambda_{Kp} > 0$ has size $(Kp \times Kp)$. The regularization structure imposed is different in general than that in $\hat{\bm{\beta}}^R(\Lambda)$: $\Lambda_{Kp}$ induces \textit{column-wise} ridge regularization, which penalizes coefficient estimates uniformly over each of the $Kp$ columns of $B$. 
The associated vectorized estimator then simplifies:
\begin{align*}
	\hat{\bm{\beta}}^R(\Lambda_{Kp}) 
	& = \left( \left( \frac{Z Z'}{T} + \Lambda_{Kp} \right) \otimes I_K \right)\inv (Z \otimes I_K)\bof{y} \\
	& = \left( \left( \frac{Z Z'}{T} + \Lambda_{Kp} \right)\inv \frac{Z \otimes I_K}{T} \right)\bof{y}
\end{align*}
On the other hand, the \textit{devectorized} RLS estimator is given by
\begin{equation*}
	\hat{B}^R(\Lambda_{K^2p}) = \text{reshape}(\bm{\beta}^R(\Lambda), K, Kp)
\end{equation*}
that is, $\hat{B}^R$ is simply a restructuring of the vectorized estimator into a matrix with identical dimensions to $B$. Importantly then, $\hat{B}^R(\Lambda_{K^2p})$ is equivalent to $\hat{B}^R(\Lambda_{Kp})$ if $\Lambda_{K^2p} = \Lambda_{Kp} \otimes I_K$. Because $\bm{\beta}^R(\Lambda_{K^2p})$ and $\hat{B}^R(\Lambda_{K^2p})$ allow for the most generality in penalization structure, I will consider them to be the reference RLS estimators, so the dimension subscript to $\Lambda$ will be dropped unless explicitly required.

\subsection{Autocovariance and Asymptotic Conditioning}\label{Autocov_Matrix_Conditioning}

The conditioning of the autocovariance $\Sigma_y = \E[ y_t y_t']$ is an important measure for the role that the regularization in the RLS estimator should be playing. This in turn depends on the eigenvalues of $\hat{\Sigma}_y$ with respect to those of $\Sigma_y$. \cite{hoerlApplicationNonorthogonal1970} showed in the linear regression setting that, when the sample covariance matrix deviates significantly from the identity matrix, its small eigenvalues excessively inflate the variance of least squares estimates, even though the regression problem itself is well-posed. This fragility is inherently a byproduct of finite sampling, and partially due to numerical procedures. Nowadays, numerical precision is virtually not a concern anymore, as robust linear algebra procedures are implicitly implemented in most scientific languages and toolboxes. Yet estimation issues tied to small or unfavorable data samples remain extremely relevant from both theoretical and practical viewpoints. 

In the spirit of ridge as a regularization procedure, the following Lemma establishes convergence in probability of the ordered eigenvalues of the sample autocovariance matrix.

\begin{lemma_apx}
	If $\hat{\Sigma}_y = T\inv \sum_{t=1}^{T-1} y_t y_t' \toprob \Sigma_y$ where $\Gamma \in \Real^{K\times K}$ is positive definite, then 
	\begin{equation*}
		\omega_j\left( \hat{\Sigma}_y \right) \toprob \omega_j(\Sigma_y)
	\end{equation*}
	where $\omega_j(A)$ is the $j$ largest eigenvalue of $A$.
\end{lemma_apx}
\begin{proof}
	
	First, recall that for all matrices $A \in \Real^{K\times K}$, the determinant $\det(A)$ is clearly a continuous mapping\footnote{This follows from $\det(A_{i,j}) = \sum_{\varsigma} \text{sign}(\varsigma) \prod_{i=1}^{K} A_{\varsigma(i),i} $ for permutation $\varsigma$ over $\{1,\ldots,K\}$}. Furthermore, for any polynomial $g(z) = z^n + a_1 z^{n-1} + \ldots + a_n$, $a_i \in \Complex$ factored as $g(z) = (z - w_1) \cdots (z - w_n)$, $w_i \in \Complex$, where the ordering of roots $w_i$ is arbitrary, it holds that for any $\epsilon > 0$ there exists $\delta > 0$ such that for every polynomial $h(z) =  z^n + b_1 z^{n-1} + \ldots + b_n$ with $\abs{a_i - b_i} < \delta$ decomposed as $g(z) = (z - \overline{w}_1) \cdots (z - \overline{w}_n)$, $\abs{w_i - \overline{w}_i} < \epsilon$, $i = 1, \ldots, n$, see \cite{whitney1972complex}, Appendix V.4. This in particular implies that the roots of the characteristic polynomial of matrix $A$ are continuous functions of its coefficients.
	
	Let $\varrho_{\hat{\Sigma}_y}(z) = z^K + a_1 z^{K-1} + \ldots + a_K = (z - \hat{\omega}_{1}) \cdots (z - \hat{\omega}_{K})$ and $\varrho_{\Sigma_y}(z) = z^K + b_1 z^{K-1} + \ldots + b_K = (z - {\omega}_{1}) \cdots (z - {\omega}_{K})$ be the (real) characteristic polynomials of $\hat{\Sigma}_y$ and $\Sigma_y$ respectively. Because of the continuity arguments above, for every $\epsilon > 0$ there exist $\delta_1$, $\delta_2 > 0$ such that
	\begin{align*}
		\Prob(\abs{\hat{\omega}_{i} - \omega_{i}} > \epsilon) 
		& \leq \Prob(\abs{a_i - b_i} > \delta_1) \\
		& \leq \Prob(\norm{\hat{\Sigma}_y - \Sigma_y} > \delta_2)
	\end{align*}
	for $i \in \{1, \ldots, K\}$. Since by assumption $\hat{\Sigma}_y \toprob \Sigma_y$, the RHS of the above converges to zero as $T \to \infty$, thus $\hat{\omega}_i \toprob \omega_i$.
\end{proof}

\section{Proofs}\label{appendix:proofs}

\subsection{Shrinkage}\label{Appendix_Shrinkage_Proofs}

%

\paragraph{Proof of Proposition \ref{prop:lag_adapt_shrink}}

\begin{proof}
	Notice that, by introducing $\Lambda_p := \text{diag}\{\lambda_1, \ldots, \lambda_p\}$, any lag-adapted regularization matrix can be written as $\Lambda^{(p)} = \Lambda_p \otimes I_{K^2} = (\Lambda_{p} \otimes I_K) \otimes I_K$, so that
	\begin{align*}
		\hat{\bm{\beta}}^R(\Lambda^{(p)}_i) 
		& = \left[ (Z Z' + \Lambda_{p,i} \otimes I_K) \otimes I_K \right]\inv (Z \otimes I_K)\bof{y} \nonumber \\ 
		& = \left[ (Z Z' + \Lambda_{p,i} \otimes I_K)\inv \otimes I_K \right] (Z \otimes I_K)\bof{y}
	\end{align*}
	by the properties of Kronecker product. It is now possible to derive the statements of the proposition as follows:
	\begin{itemize}
		\item[(a)] The result regarding isotropic regularizer $\Lambda^{(p)} = \lambda I_{K^2 p}$ is trivial given Proposition \ref{prop:iso_shrink}. 
		\item[(b)] Without loss of generality due to the ordering of lags in $Z$, one may write the Gram matrix $Z Z'$ in a block fashion,
		\begin{equation*}
			Z Z' + \Lambda_p = 
			\begin{bmatrix}
				(Z Z')_{[\mathcal{S}]} + \Lambda_{[\mathcal{S}]} & D \\
				D' & (Z Z')_{[\mathcal{S}^\mathnormal{c}]} + \Lambda_{[\mathcal{S}^\mathnormal{c}]} 
			\end{bmatrix}
		\end{equation*}
		where $(Z Z')_{[\mathcal{S}^\mathnormal{c}]}$ is the sub-matrix containing all the components \textit{not} indexed by subset $\mathcal{S}$, and the subscript has been dropped from $\Lambda_p$ for ease of notation. \\
		Define $A_{[\mathcal{S}]} = (Z Z')_{[\mathcal{S}]} + \Lambda_{[\mathcal{S}]}$, $B_{[\mathcal{S}^\mathnormal{c}]} = (Z Z')_{[\mathcal{S}^\mathnormal{c}]} + \Lambda_{[\mathcal{S}^\mathnormal{c}]}$ and $\varDelta = (B_{[\mathcal{S}^\mathnormal{c}]} - D'A_{[\mathcal{S}]}\inv D)$. The matrix block-inversion formula yields
		\begin{equation*}
			(Z Z' + \Lambda_p)\inv =
			\begin{bmatrix}
				A_{[\mathcal{S}]}\inv + A_{[\mathcal{S}]}\inv D \varDelta\inv D' A_{[\mathcal{S}]}\inv & A_{[\mathcal{S}]}\inv D \varDelta\inv \\
				- \varDelta\inv D' A_{[\mathcal{S}]}\inv & \varDelta\inv 
			\end{bmatrix}
			.
		\end{equation*}
		If $\Lambda_{[\mathcal{S}]} \to 0$ and $\Lambda_{[\mathcal{S}^\mathnormal{c}]} \to \infty$, then $A_{[\mathcal{S}]} \to (Z Z')_{[\mathcal{S}]}$, $B_{[\mathcal{S}^\mathnormal{c}]} \to \infty$. Therefore $\varDelta\inv \to 0$, since for $\Lambda_{[\mathcal{S}^\mathnormal{c}]}$ sufficiently large $\norm{B_{[\mathcal{S}^\mathnormal{c}]}\inv D'A_{[\mathcal{S}]}\inv D} < 1$ and thus the Sherman-Morrison-Woodbury formula implies
		\begin{equation*}
			\norm{ (B_{[\mathcal{S}^\mathnormal{c}]} - D'A_{[\mathcal{S}]}\inv D)\inv } \leq \frac{\norm{B_{[\mathcal{S}]}\inv}}{1 - \norm{B_{[\mathcal{S}^\mathnormal{c}]}\inv D'A_{[\mathcal{S}]}\inv D}} \to 0.
		\end{equation*}
		The above results finally yield 
		\begin{equation*}
			\left[ (Z Z' + \Lambda_p)\inv Z \otimes I_K \right] \bm{y} \to 
			\begin{bmatrix}
				(Z Z')_{[\mathcal{S}]} & 0 \\
				0 & 0
			\end{bmatrix}
			(Z \otimes I_K) \bm{y}
			= \hat{\bm{\beta}}^{LS}_{[\mathcal{S}]}
		\end{equation*}
		as required.
	\end{itemize}
\end{proof}

\subsection{Ridge Asymptotic Theory}\label{Appendix_Asymptotic_Proofs}

\paragraph{Proof of Theorem~\ref{th:MAIN_THEOREM}}

\begin{proof}
	\begin{itemize}
		\item[(a)] Assumptions A-B imply directly that $\hat{\Gamma}_T$ is a consistent estimator for $\Gamma$: in particular, $y_t$ is a stationary, stable and ergodic VAR process.
		\item[(b)] Rewriting $\hat{\bm{\beta}}^R(\Lambda, \bm{\beta}_0)$ yields
		\begin{align*}
			\hat{\bm{\beta}}^R(\Lambda, \bm{\beta}_0)
			& = \left( \frac{Z Z'}{T} \otimes I_K + \Lambda \right)\inv \left[ T\inv (Z \otimes I_K)((Z' \otimes I_K) \bm{\beta} + \bof{u}) + \Lambda \bm{\beta}_0 \right] \\
			& = \left( \frac{Z Z'}{T} \otimes I_K + \Lambda \right)\inv \left[ \left( \frac{Z Z'}{T} \otimes I_K \right) \bm{\beta} + \frac{(Z \otimes I_K)\bof{u}}{T} + \Lambda \bm{\beta}_0 \right] \\
			& = \left( \frac{Z Z'}{T} \otimes I_K + \Lambda \right)\inv \left[ \left( \frac{Z Z'}{T} \otimes I_K + \Lambda \right) \bm{\beta} + \frac{(Z \otimes I_K)\bof{u}}{T} + \Lambda (\bm{\beta}_0 - \bm{\beta}) \right] \\
			& = \bm{\beta} + \left( \frac{Z Z'}{T} \otimes I_K + \Lambda \right)\inv  \Lambda (\bm{\beta}_0 - \bm{\beta}) + \left( \frac{Z Z'}{T} \otimes I_K + \Lambda \right)\inv \frac{(Z \otimes I_K)\bof{u}}{T} .
		\end{align*}
		I study the last two terms of the last equality separately. The first term is $o_p(1)$,
		\begin{equation*}
			\left( \left(\frac{Z Z'}{T}\right) \otimes I_K + \Lambda \right)\inv \Lambda (\bm{\beta}_0 - \bm{\beta}) 
			= \left( \left(\frac{Z Z'}{T}\right) \otimes I_K + o_p(1) \right)\inv o_p(1) (\bm{\beta}_0 - \bm{\beta}) 
			\toprob {0}
		\end{equation*}
		since $(\bm{\beta} - \bm{\beta}_0) = (\bm{\beta} - \underline{\bm{\beta}}_0) + (\underline{\bm{\beta}}_0 - \bm{\beta}_0) = (\bm{\beta} - \underline{\bm{\beta}}_0) + o_p(1)$. Considering the matrix sequence
		\begin{equation*}
			\bm{\zeta}_T = \left[ T\inv (Z Z'), T\inv \Lambda \right] ,
		\end{equation*}
		which under Assumptions B and D.(ii) converges in probability to $[\Gamma, 0]$, by the continuous mapping theorem \citep{davidson1994stochastic} the second term gives
		\begin{equation*}
			\left( \left(\frac{Z Z'}{T}\right) \otimes I_K + o_p(1) \right)\inv \left(\frac{1}{T}(Z \otimes I_K)\bof{u} \right) \toprob \Gamma\inv \, \E[(Z \otimes I_K)\bof{u}] = {0}
		\end{equation*}
		under Assumption A.
		\item[(c)] The residuals $\hat{U}$ can be written as
		\begin{equation*}
			\hat{U} = Y - \hat{B}^R Z = B Z + U - \hat{B}^R Z = U + (B - \hat{B}^R) Z
		\end{equation*}
		Thus 
		\begin{equation}
			\frac{\hat{U} \hat{U}'}{T} = \frac{U U'}{T} +  (B - \hat{B}^R) \left(\frac{Z Z'}{T}\right) (B - \hat{B}^R)' + (B - \hat{B}^R) \left(\frac{Z U'}{T}\right) + \left(\frac{U Z'}{T}\right) (B - \hat{B}^R)'
		\end{equation}
		From (a) one finds that $\vect(B) - \vect(\hat{B}^R) = \bm{\beta} - \hat{\bm{\beta}}^R = o_p(1)$, so $(B - \hat{B}^R) \toprob 0$, while $T\inv (Z Z') \toprob \E[Z Z']$ and $T\inv (Z U') \toprob \E[Z U'] = 0$: the terms involving these quantities then vanish asymptotically. Lastly, the first term of the sum gives
		\begin{equation*}
			\frac{U U'}{T} = \frac{1}{T} \sum_{t=1}^T u_t u_t' \toprob \E[u_t u_t'] = \Sigma_u
		\end{equation*}
		for $T \to \infty$ under Assumptions A and B.
		\item[(d)] With the same expansion used in (b), 
		\begin{equation*}
			\sqrt{T}(\hat{\bm{\beta}}^R(\Lambda, \bm{\beta}_0) - \bm{\beta}) 
			= 
			Q_T\inv \left( \sqrt{T} \Lambda \right) (\bm{\beta}_0 - \bm{\beta}) + Q_T\inv \left(\frac{1}{\sqrt{T}}(Z \otimes I_K)\bof{u} \right)
		\end{equation*}
		where $Q_T = (T\inv Z Z' + \Lambda) \toprob \Gamma$. Following the arguments above, the first term in the sum converges in probability,
		\begin{equation*}
			Q_T\inv \left( \sqrt{T} \Lambda \right) \left( \underline{\bm{\beta}}_0 - \bm{\beta} + o_p(1) \right) 
			\toprob 
			\Gamma\inv \Lambda_0 (\bm{\beta} - \underline{\bm{\beta}}_0)
		\end{equation*}
		The second term has normal limiting distribution,
		\begin{equation*}
			Q_T\inv \left(\frac{1}{\sqrt{T}}(Z \otimes I_K)\bof{u} \right) \todist \mathcal{N}(0, \Gamma\inv \otimes \Sigma_u)
		\end{equation*}
		see \cite{lutkepohl2005new}, Proposition 3.1. By Slutky's theorem claim (d) follows.
	\end{itemize}
\end{proof}

\paragraph{Proof of Theorem~\ref{th:Corollary_MAIN_THEOREM}}

\begin{proof}
	\begin{itemize}
		\item[(1)] Since condition (i) implies that $\sqrt{T} \Lambda \toprob 0$, results (a)-(c) are unchanged, while (d) now involves the limit
		\begin{equation*}
			Q_T\inv \left( \sqrt{T} \Lambda \right) (\bm{\beta}_0 - \bm{\beta}) 
			= 
			Q_T\inv \cdot o_P(1) \cdot \left( \underline{\bm{\beta}}_0 - \bm{\beta} + o_P(1) \right) \toprob 0
		\end{equation*}
		yielding (d$'$).
		\item[(2)] Assuming $\bm{\beta}_0 \toprob \bm{\beta}$ simplifies the terms in the proof of Theorem \ref{th:MAIN_THEOREM} since now $\bm{\beta} - \bm{\beta}_0 = o_P(1)$. The weaker rate imposed on $\Lambda$ does not influence results (a)-(c). Moreover,
		\begin{equation*}
			Q_T\inv \cdot ( \sqrt{T} \Lambda ) \cdot ( \bm{\beta}_0 - \bm{\beta} ) = Q_T\inv \cdot O_P(1) \cdot o_P(1) \toprob 0
		\end{equation*}
		so that (d$'$) follows.
	\end{itemize}
\end{proof}

\paragraph{Proof of Theorem~\ref{Joint_MAIN_THEOREM}}

\begin{proof}
	
	I make a straightforward adaptation of the proof found in \cite{hamilton1994time}, Proposition 11.2. Define $\hat{\Sigma}_u^* = T\inv (U U')$, which is expanded to
	\begin{align*}
		\hat{\Sigma}_u^* 
		& = \frac{1}{T} (Y - B Z) (Y - B Z)' \\
		& = \frac{1}{T} \left(Y - \hat{B}^R Z + (\hat{B}^R - B) Z \right) \left(Y - \hat{B}^R Z + (\hat{B}^R - B) Z \right)' \\
		& = \hat{\Sigma}_u^R + (\hat{B}^R - B) \left(\frac{Z Z'}{T}\right) (\hat{B}^R - B)' + \nonumber \\
		& \qquad + \frac{1}{T} \left( (Y - \hat{B}^R Z) Z' (\hat{B}^R - B)' + (\hat{B}^R - B) Z (Y - \hat{B}^R Z)' \right)
	\end{align*}
	Contrary to the least squares estimator, cross-terms do not cancel out since for $\Lambda \succeq 0$ the residuals $(Y - \hat{B}^R Z)$ are not in the orthogonal space of $Z$. From the consistency results of Theorem \ref{th:MAIN_THEOREM},
	\begin{equation*}
		(\hat{B}^R - B) \left(\frac{Z Z'}{T}\right) (\hat{B}^R - B)' = o_p(1) \left(\frac{Z Z'}{T}\right) o_p(1) \toprob 0
	\end{equation*}
	and
	\begin{equation*}
		\sqrt{T} (\hat{B}^R - B) \left(\frac{Z Z'}{T}\right) (\hat{B}^R - B)' = O_p(1) \left(\frac{Z Z'}{T}\right) o_p(1) \toprob 0
	\end{equation*}
	Further,
	\begin{equation*}
		\sqrt{T} \left[\frac{1}{T} (Y - \hat{B}^R Z) Z' (\hat{B}^R - B)'\right] = \left(\frac{\hat{U} Z'}{T}\right) \sqrt{T} (\hat{B}^R - B)' \toprob 0 
	\end{equation*}
	since again $\sqrt{T} \hat{B}^R$ is asymptotically normal, and $T\inv (\hat{U} Z') = T\inv (U Z') + (B - \hat{B}^R) \cdot T\inv (Z Z') = T\inv (U Z') + o_p(1) \toprob \E[U Z'] = 0$. The same holds for the remaining transpose term, too.
	
	It is thus proven that $\sqrt{T} (\hat{\Sigma}_u^* - \hat{\Sigma}_u^R) \toprob 0$, meaning $\sqrt{T} (\hat{\Sigma}_u^* - \Sigma_u) \toprob \sqrt{T} (\hat{\Sigma}_u^R - \Sigma_u)$ so that the two terms may be exchanged in computing the joint asymptotic distribution. Theorem \ref{th:MAIN_THEOREM} accordingly yields
	\begin{equation*}
		\begin{bmatrix}
			\hat{\bm{\beta}}^R_0 - \bm{\beta} \\
			\vect(\hat{\Sigma}_u^R) - \vect(\Sigma_u)
		\end{bmatrix}
		\toprob
		\begin{bmatrix}
			Q_T\inv \left( \sqrt{T} \Lambda \right) (\bm{\beta}_0 - \bm{\beta}) + Q_T\inv \frac{1}{\sqrt{T}} \bm{\xi} \\
			\frac{1}{\sqrt{T}} \bm{\eta}
		\end{bmatrix}
	\end{equation*} 
	where $\bm{\xi} = (Z \otimes I_K)\bof{u}$ and $\bm{\eta} = \vect(U U' - \Sigma_u)$. As in \cite{hamilton1994time}, Proof of Proposition 11.2, $(\bm{\xi}', \bm{\eta})'$ is a martingale difference sequence, thus the claim
	\begin{equation*}
		\sqrt{T} \begin{bmatrix}
			\hat{\bm{\beta}}^R_0 - \bm{\beta} \\
			\vect(\hat{\Sigma}_u^R) - \vect(\Sigma_u)
		\end{bmatrix}
		\todist 
		\mathcal{N}\left(
		\begin{bmatrix}
			\Gamma\inv \Lambda_0 (\underline{\bm{\beta}}_0 - \bm{\beta}) \\
			0
		\end{bmatrix},
		\begin{bmatrix}
			\Gamma\inv \otimes \Sigma_u & 0 \: \\
			0 & \Omega \:
		\end{bmatrix}
		\right)
	\end{equation*}
	as $T \to \infty$ follows. When the strengthened assumptions (1) or (2) of Theorem \ref{th:Corollary_MAIN_THEOREM} are used instead, the non-zero limiting mean vanishes
	\begin{equation*}
		Q_T\inv \cdot \Lambda \cdot \sqrt{T} (\bm{\beta}_0 - \bm{\beta}) \toprob 0
	\end{equation*}
	proving that the joint asymptotic distribution is mean-zero Gaussian. 
	
	Finally, to compute the explicit expression of the asymptotic variance $\Omega$, one must take care and note that $u_t$ is not assumed to be normally distributed, contrary to the remainder of the proof in \cite{hamilton1994time}, pp. 342-343. A correct expression for i.i.d. non-Gaussian $u_t$ can be found in Remark 2.1, \cite{bruggemannInferenceVARsConditional2016}, yielding
	\begin{equation*}
		\Omega = \text{Var}[\vect(u_t u_t')] = \E\big[\vect(u_t u_t')\, \vect(u_t u_t')'\big] - \bm{\sigma} \bm{\sigma}'
	\end{equation*}
	where $\bm{\sigma} = \vect(\Sigma_u)$.
\end{proof}

\paragraph{Proof of Theorem~\ref{th:MAIN_TEOREM_with_shrink}}

\begin{proof}
	The stated results reduce to studying the behavior of two components used in the proof of Theorem \ref{th:MAIN_THEOREM} and Theorem \ref{th:Corollary_MAIN_THEOREM}, under the additional simplification of $\bm{\beta}_0 = 0$. 
	\begin{itemize}
		\item[(a)] Identical to result (a) in Theorem \ref{th:MAIN_THEOREM}.
		\item[(b)] Consistency follows immediately by the fact that ${\Lambda}\, \bm{\beta} \toprob 0$.
		%
		%
		\item[(c)] Follows from (c), Theorem \ref{th:MAIN_THEOREM} and (b) above.
		\item[(d$\,''$)] The bias term in the expression of $\sqrt{T}(\hat{\bm{\beta}}^R - \bm{\beta})$ is driven by
		\begin{equation*}
			\sqrt{T} {\Lambda}\, \bm{\beta} 
			=
			\begin{bmatrix}
				\sqrt{T} \Lambda_1 \cdot \bm{\beta}_1 \\
				\sqrt{T} \Lambda_2 \cdot T^{-(1/2 + \delta)} \bm{b}_2
			\end{bmatrix} =
			\begin{bmatrix}
				o_P(1) \cdot \bm{\beta}_1 \\
				\Lambda_2 \cdot T^{-\delta} \bm{b}_2
			\end{bmatrix} \toprob 0
		\end{equation*}
		meaning there is no asymptotic bias. On the other hand, 
		\begin{equation*}
			\left( \left({T}\inv {Z Z'} \right) \otimes I_K + \Lambda \right)\inv \toprob (\Gamma + \overline{\Lambda})\inv \otimes I_K .
		\end{equation*}
		Setting $(\Gamma + \overline{\Lambda}) = \Gamma_{\overline{\Lambda}}$ yields the claim since there are no further simplifications in the asymptotic variance formula, cf. proof of (d), Theorem \ref{th:MAIN_THEOREM}.
	\end{itemize}
\end{proof}

\section{Cross-validation}\label{appendix:cross_validation}

Later in this section, the following lemma will be useful.

\begin{lemma}\label{lemma:exp_decay_MA}
	Let $y_t$ a stationary and stable mean-zero AR($p$) process with companion form matrix $\mathbb{A} \in \Real^{p \times p}$. Then, the associated MA($\infty$) coefficients, $\{\vartheta_k\}_{k \in \mathbb{N}_0}$, decay exponentially for $k$ sufficiently large, that is,
	\begin{equation*}
		\abs{\vartheta_k} = O\big( \exp(- C_\mathbb{A}\, k) \big)
	\end{equation*}
	for some constant $C_\mathbb{A} > 0$.
\end{lemma}

\begin{proof}
	Recall from e.g. \cite{lutkepohl2005new} that if $\mathbb{A}$ is the companion matrix of the AR($p$) model, then $\vartheta_k = \iota \mathbb{A}^j  \iota'$, where $\iota := (1,0, \ldots, 0)' \in \Real^p$. Note that $\norm{\iota} = 1$ by construction and that the spectral radius of $\mathbb{A}$ is less than one under the assumption of stability. Let $\mathbb{A} = {V} {D} {V}\inv$ be the Jordan canonical form of the companion matrix and $\omega_1$ the dominant eigenvalue. Under the assumption of stability $\abs{\omega_1} < 1$. Now observe that, supposing $D$ has $l \leq p$ diagonal blocks, for $k \geq 0$
    \begin{equation*}
        \frac{\abs{ \iota \mathbb{A}^k \iota' }}{\abs{\omega_1^k}}
        = 
        \abs{ \iota {V} \left( \frac{{D}^k}{\omega_1^k} \right) {V}\inv \iota' } 
        = 
        \abs{ \iota {V} \begin{bmatrix}
                [1] & & & \\
                & \left[ \frac{1}{\omega_1^k} {D}_2^k \right] & & \\
                & & \ddots & \\
                & & & \left[ \frac{1}{\omega_1^k} {D}_l^k \right]
        \end{bmatrix} {V}\inv \iota' } ,
    \end{equation*}
    where $[1]$ is the dominant Jordan block, while ${D}_2, \ldots, {D}_l$ are the non-dominant blocks.
    Then, letting $k \to \infty$, one gets
	\begin{equation*}
		\frac{\abs{ \iota \mathbb{A}^k \iota' }}{\abs{\omega_1^k}}
		\to 
		\abs{ \iota {V} \begin{bmatrix}
                [1] & & & \\
                & 0 & & \\
                & & \ddots & \\
                & & & 0
        \end{bmatrix} {V}\inv \iota' } 
        =
        C_D 
	\end{equation*}
	for some constant $C_D < \infty$, as for $\ell \in \{2, \ldots, l\}$ one can see that ${D}_\ell^k / {\omega_1^k} \to 0$. Since for complex $\omega_1$ it holds $\abs{\omega_1^k} = \abs{\omega_1}^k$, letting $C_\mathbb{A} = -\log(\abs{\omega_1})$ the proof is complete.
\end{proof}

\subsection{Two-fold CV}

\paragraph{Proof of Theorem~\ref{th:cv2_m_asconv_err}}

\begin{proof}
	Write $\textnormal{cv2}_m(\Lambda)$ as
	\begin{align*}
		\textnormal{cv2}_m(\Lambda)
		& = \widetilde{T}\inv \left( \widetilde{\bof{y}} - \widetilde{Z}' \hat{\bm{\beta}}^R(\Lambda) \right)' \left( \widetilde{\bof{y}} - \widetilde{Z}' \hat{\bm{\beta}}^R(\Lambda) \right) \\
		& = \widetilde{T}\inv \left[ \widetilde{\bof{u}} + \widetilde{Z}' \left( \bm{\beta} - \hat{\bm{\beta}}^R(\Lambda) \right) \right]' \left[ \widetilde{\bof{u}} + \widetilde{Z}' \left( \bm{\beta} - \hat{\bm{\beta}}^R(\Lambda) \right) \right] \\
		& = \left( \bm{\beta} - \hat{\bm{\beta}}^R(\Lambda) \right)' \left( \frac{\widetilde{Z} \widetilde{Z}'}{\widetilde{T}} \right) \left( \bm{\beta} - \hat{\bm{\beta}}^R(\Lambda) \right)
			+ 2 \left( \bm{\beta} - \hat{\bm{\beta}}^R(\Lambda) \right)' \left( \frac{\widetilde{Z} \widetilde{\bof{u}}}{\widetilde{T}} \right)
			+ \frac{\widetilde{\bof{u}}' \widetilde{\bof{u}}}{\widetilde{T}} .
	\end{align*}
	By a strong LLN for weakly dependent processes (see, for example, \citealt{Rio2017}), it holds that $\widetilde{T}\inv \widetilde{Z} \widetilde{Z}' \overset{a.s.}{\longrightarrow} \Gamma_z $, $\widetilde{T}\inv \widetilde{Z} \widetilde{\bof{u}} \overset{a.s.}{\longrightarrow} \bof{0}$ and $\widetilde{T}\inv \widetilde{\bof{u}}' \widetilde{\bof{u}} \overset{a.s.}{\longrightarrow} \Sigma_u$.
	Since, by a similar decomposition, it holds that 
	\begin{equation*}
		\textnormal{err}(\Lambda) 
		= 
		\left( \bm{\beta} - \hat{\bm{\beta}}^R(\Lambda) \right)' \Gamma \left( \bm{\beta} - \hat{\bm{\beta}}^R(\Lambda) \right)
		+ \Sigma_u ,
	\end{equation*}
	where $\Gamma = \E[z_t z_t']$ is positive definite, almost sure convergence is proven.
	
	To prove uniform convergence over compact subsets, I follow the proof of \cite{Patil2021}, Theorem 4.1, which relies on verifying the conditions of the Arzel\`a-Ascoli theorem. That is, one must prove that function $\textnormal{cv2}_m(\Lambda)$ as well as its first derivatives are bounded over compact sets.
	As the Arzel\`a-Ascoli theorem readily generalizes to Euclidean spaces of arbitrary (fixed) dimension, I will directly consider the matrix derivative when checking boundedness. 
	
	Assume that $\Lambda \in \mathcal{I}$, where $\mathcal{I}$ is a compact set of positive semidefinite penalization matrices $\Lambda$ such that $\norm{\Lambda}_{\max} < \infty$ and $\Lambda \succ \lambda_{\min} I$.
	Note $\bm{\beta} - \hat{\bm{\beta}}^R(\Lambda) = - (\hat{\Gamma} + \Lambda)\inv \Lambda \bm{\beta} + (\hat{\Gamma} + \Lambda)\inv (T\inv Z \bof{u})$ where $\hat{\Gamma} = T\inv Z Z'$. Using this decomposition, one gets first
	\begin{align*}
		\abs{ \textnormal{err}(\Lambda) }
		& \leq \Sigma_u + \norm{ \Gamma }_2 \norm{ \bm{\beta} - \hat{\bm{\beta}}^R(\Lambda) }_2^2  \\
		%
		%
		& \leq \Sigma_u  + \norm{ \Gamma }_2 \left[ \norm{ (\hat{\Gamma} + \Lambda)\inv \Lambda }_2^2 \norm{ \bm{\beta} }_2^2 + \norm{ (\hat{\Gamma} + \Lambda)\inv }_2^2 \norm{ T\inv Z \bof{u} }_2^2  \right. \\
		& \qquad\qquad\qquad\quad \left. + 2 \norm{ (\hat{\Gamma} + \Lambda)\inv }_2^2 \norm{ T\inv Z \bof{u} }_2 \norm{ \Lambda \bm{\beta} }_2 \right] \\
		& \overset{a.s.}{\leq}\: \Sigma_u + \norm{ \Gamma }_2 \frac{ \omega_{\max}(\Lambda)^2 \norm{\bm{\beta}}_2^2 + C_{zu}^2 + 2 C_{zu}\, \omega_{\max}(\Lambda) \norm{\bm{\beta}}_2 }{( \omega_{\min}(\hat{\Gamma}) + \omega_{\min}(\Lambda) )^2}
	\end{align*}
	where the last line follows from applying Weil's eigenvalue inequalities \citep{Bhatia1997} to $(\hat{\Gamma} + \Lambda)\inv$ and the fact that $T\inv Z \bof{u} \overset{a.s.}{\longrightarrow} \bof{0}$ by a strong LLN, so that there exists a constant $C_{zu} > 0$ bounding $\norm{ T\inv Z \bof{u} }_2$ for $T$ large enough.
	
	Additionally, the matrix derivative of $\bm{\beta} - \hat{\bm{\beta}}^R(\Lambda)$ with respect to $\Lambda$ is
	\begin{equation*}
		\frac{\partial (\bm{\beta} - \hat{\bm{\beta}}^R(\Lambda))}{\partial \Lambda}
		= 
		- (\hat{\Gamma} + \Lambda)\inv \bm{\beta} + (\hat{\Gamma} + \Lambda)^{-2} \Lambda \bm{\beta} + (\hat{\Gamma} + \Lambda)^{-2}(T\inv Z \bof{u}) ,
	\end{equation*}
	so that, by using similar argument as the one used above, one gets
	\begin{align*}
		\left\lvert \frac{\partial\, \textnormal{err}(\Lambda) }{\partial \Lambda} \right\rvert
		& \:\overset{a.s.}{\leq}\:
		\Sigma_u + \norm{ \Gamma }_2 \left[ 
			\frac{ \norm{\bm{\beta}}_2^2 }{( \omega_{\min}(\hat{\Gamma}) + \omega_{\min}(\Lambda) )^2} 
			+ \frac{ 2\, \omega_{\max}(\Lambda) \norm{\bm{\beta}}_2^2 + 2 C_{zu}\, \norm{\bm{\beta}}_2 }{( \omega_{\min}(\hat{\Gamma}) + \omega_{\min}(\Lambda) )^3} \right. \\
		& \qquad\qquad\qquad\qquad \left.
			+ \frac{ \omega_{\max}(\Lambda)^2 \norm{\bm{\beta}}_2^2 + C_{zu}^2 + 2 C_{zu}\, \omega_{\max}(\Lambda) \norm{\bm{\beta}}_2 }{( \omega_{\min}(\hat{\Gamma}) + \omega_{\min}(\Lambda) )^4}
		\right] .
	\end{align*}
	The almost sure bound in the last display is also clearly finite for any $\Lambda \in \mathcal{I}_\lambda$, as required. 
	
	One can easily bound $\textnormal{cv2}_m(\Lambda)$ and its first derivative as $\textnormal{err}(\Lambda)$, with only addition of an extra term depending on $ ( \bm{\beta} - \hat{\bm{\beta}}^R(\Lambda) )' ( {\widetilde{T}}\inv {\widetilde{Z} \widetilde{\bof{u}}} )$. This means that $\textnormal{err}(\Lambda) - \textnormal{cv2}_m(\Lambda)$ forms an equicontinuous family of functions with respect to $\Lambda$ over any $\mathcal{I}_\lambda$. Therefore, Arzel\`a-Ascoli yields uniform convergence of a subsequence, and since the difference converges to zero pointwise, too, the entire sequence converges uniformly.
\end{proof}

\subsection{Cross-validation under Dependence}\label{section:cv_under_dependence}

The result of Theorem \ref{th:cv2_m_asconv_err} may be only partially informative in practice, as it does not give information on how dependence, in terms of the buffer block of size $m$, impacts $\textnormal{cv2}_m(\Lambda)$. Indeed, due to averaging, the effects of time dependence between the estimation and evaluation folds are washed out in the limit $\widetilde{T} \to \infty$ even when $m$ is fixed. Therefore, Theorem \ref{th:cv2_m_asconv_err} is not useful in finite samples, where one would preferably set $m$ to be as small as possible.

To address dependence, in the same setup as above, consider an alternative predictive error measure, the \textit{$m$-dependence prediction error}, 
\begin{equation*}\label{eq:m_depend_pred_err}
	\textnormal{Err}_m\left( \hat{\bm{\beta}}^R(\Lambda) \right) 
	:= 
	\E_{\,{y}_{T+m+1}, {z}_{T+m+1}}\left[ \left({y}_{T+m+1} - {z}_{T+m+1}'\, \hat{\bm{\beta}}^R(\Lambda) \right)^2 \big\vert Z, \bof{y} \right] ,
\end{equation*}
and the associated error curve, $\textnormal{err}_m(\Lambda) := \textnormal{Err}_m\left( \hat{\bm{\beta}}^R(\Lambda) \right)$. The empirical counterpart to this quantity is given by $\textnormal{cv2}_m(\Lambda)$ for $\widetilde{T} = 1$.
The next theorem shows that in the case of a purely autoregressive data generating process, the error one commits by choosing a finite buffer size is exponentially small.

\begin{theorem}\label{th:cv_dependence}
	Under Assumptions A-C, for every $\Lambda$ in the cone of diagonal positive definite penalty matrices with diagonal entries in $(\lambda_{\min}, \infty)$, it holds that
	\begin{equation*}
		\textnormal{err}_m(\Lambda) - \textnormal{err}(\Lambda) 
		=
		O(\exp(- C_{\bm{\beta}}\, m))
	\end{equation*}
	where $C_{\bm{\beta}}$ is a constant that does not depend on $\Lambda$.
\end{theorem}

\begin{proof}
	In line with the definition of $\textnormal{cv2}_m(\Lambda)$, I set $\widetilde{y}_1 = {y}_{T+m+1}$, $\widetilde{z} = {z}_{T+m+1}$ and $\widetilde{u}_1 = {y}_{T+m+1} - {z}_{T+m+1}' \bm{\beta}$.
	With the same approach as in the proof of Theorem \ref{th:cv2_m_asconv_err}, here one finds
	\begin{align*}
		\textnormal{err}_m(\Lambda) 
		& = \left( \bm{\beta} - \hat{\bm{\beta}}^R(\Lambda) \right)' \E\left[ \widetilde{z}_{1} \widetilde{z}_{1}' \big\vert Z, \bof{y} \right] \left( \bm{\beta} - \hat{\bm{\beta}}^R(\Lambda) \right) \\
		& \qquad + 2 \left( \bm{\beta} - \hat{\bm{\beta}}^R(\Lambda) \right)' \E\left[ \widetilde{z}_{1} \widetilde{u}_{1} \big\vert Z, \bof{y} \right] + \E\left[ \widetilde{u}_{1}^2 \big\vert Z, \bof{y} \right] ,
	\end{align*}
	where I have removed the subscript from expectation $\E$ to make notation clearer. Since $\widetilde{u}_{1}$ is independent of $\widetilde{z}_{1}$, the cross term reduces to zero, while  
	%
		$\E\left[ \widetilde{u}_{1}^2 \big\vert Z, \bof{y} \right] = \Sigma_u$ .
	%
	Thus, it is the first term in the last display that is effected by dependence. 
	
%
	To see this, let $\widetilde{z}_{i1}$ for $1 \leq i \leq p$ be the $i$th entry of $\widetilde{z}_{1}$. Then, using the MA($\infty$) decomposition of $y_t$, i.e. $y_t = \sum_{\ell=0}^\infty \phi_\ell u_{t-\ell}$, one can write
	\begin{align*}
		\widetilde{z}_{i1} = z_{T+m+1-i} 
		& = \sum_{\ell=0}^{m-i} \phi_\ell u_{T+m+1-i-\ell} + \sum_{\ell=M+1-i}^{\infty} \phi_\ell u_{T+m+1-i-\ell} \\
		& = \sum_{\ell=0}^{m-i} \phi_\ell u_{T+m+1-i-\ell} + \sum_{s=0}^{\infty} \phi_{m + 1 -i + s} u_{T-s} \\
		& = \eta_i + \zeta_i .
	\end{align*}
	Note that $\eta_i$ is independent of $\zeta_i$, and $\zeta_i$ is belongs with the $\sigma$-algebra generated by $Z$ and $\bof{y}$. Therefore,
	\begin{align*}
		\E\left[ \widetilde{z}_{1} \widetilde{z}_{1}' \big\vert Z, \bof{y} \right]
		& = \E\left[ (\eta_i + \zeta_i) (\eta_i + \zeta_i)' \big\vert Z, \bof{y} \right] \\
		& = \Gamma_\eta + \zeta_i \zeta_i' ,
	\end{align*}
	as $\E\left[ \eta_i \zeta_i' \big\vert Z, \bof{y} \right] = \E\left[ \eta_i \big\vert Z, \bof{y} \right] \zeta_i' = 0$. 
	
	Now, I prove that $\Gamma_\eta \to \Gamma_z$ and $\zeta_i \zeta_i' \to 0$ at an exponential rate. First, let
	\begin{equation*}
		H_\phi := \begin{bmatrix}
			\sum_{s=0}^{\infty} \phi_{m + 1 + s}^2 & \sum_{s=0}^{\infty} \phi_{m + 1 + s} \phi_{m + 1 + s - 1}  & & \\
			\sum_{s=0}^{\infty} \phi_{m + 1 + s -1 } \phi_{m + 1 + s} & \sum_{s=1}^{\infty} \phi_{m + 1 + s}^2 & & \\
			& & \ddots & \\
			& & & \sum_{s=p}^{\infty} \phi_{m + 1 + s}^2 
		\end{bmatrix}
	\end{equation*}
	and observe that
	\begin{equation*}
		\norm{ \Gamma_\eta - \Gamma_z }_2 
		\leq \Sigma_u\, p \norm{ H_\phi }_{\max} 
		\leq C \exp(- C_\eta \, m) ,
	\end{equation*}
	since $\Gamma_\eta = \Sigma_u \otimes H_\phi$, lag order $p$ is fixed and the maximal entry of $H_\phi$ decays exponentially for $m$ sufficiently large following Lemma \ref{lemma:exp_decay_MA}. Secondly, much in the same vein 
	\begin{equation*}
		\norm{ \zeta_i \zeta_i' }_2 
		= \zeta_i' \zeta_i 
		\leq p \left( \sum_{s=0}^{\infty} \phi_{m + 1 -i + s} u_{T-s} \right)^2
		\leq C' \exp(- C_\zeta \, m) .
	\end{equation*}
	The proof concludes by setting $C_{\bm{\beta}} = \max(C_\eta, C_\zeta)$.
\end{proof}

\begin{remark}
	Theorem~\ref{th:cv_dependence} is reassuring because it suggests that, in practice, if the AR($p$) model is correctly specified, one may keep $m$ small and still get a valid prediction error estimate in sense of Theorem~\ref{th:cv2_m_asconv_err}. In simulations, I set $m = 0$, which is a common simplification to more effectively exploit the entire sample and does not, as discussed above, effect consistency \citep{bergmeirNoteValidityCrossvalidation2018}.
	Moreover, note that Theorem~\ref{th:cv_dependence} intuitively gives a \textit{worst-case} rate:  the dependence of between $\widetilde{z}_t$ and data in the estimation set gets milder, on average, as $\widetilde{T}$ grows. Thus, if CV aspect ratio $\widetilde{T} / T$ is balanced, dependence only plays a negligible role.
\end{remark}

\subsection{Asymptotically Valid CV}

\paragraph{Proof of Theorem \ref{th:asym_valid_cv_penalty}}

\begin{proof}
	First, recall that
	\begin{align*}
		\hat{\bm{\beta}}^R_\blacklozenge(\Lambda) 
			& = \left( \frac{Z Z'}{T} + \sqrt{T} \Lambda \right)\inv \frac{Z\bof{y}}{T} \\
			& = \bm{\beta} - \left( \frac{Z Z'}{T} + \sqrt{T} \Lambda \right)\inv \left( \sqrt{T} \Lambda \right) \bm{\beta} + \left( \frac{Z Z'}{T} + \sqrt{T} \Lambda \right)\inv \frac{Z \bof{u}}{T} .
	\end{align*}
	It also holds $\textnormal{Err}\left( \hat{\bm{\beta}}^R_\blacklozenge(\Lambda) \right) 
	= 
	\left( \bm{\beta} - \hat{\bm{\beta}}^R_\blacklozenge(\Lambda) \right)' \Gamma \left( \bm{\beta} - \hat{\bm{\beta}}^R_\blacklozenge(\Lambda) \right)
	+ \Sigma_u$. Now, notice that
	\begin{equation*}
		\left( {T}\inv Z Z' + \sqrt{T} \Lambda \right)\inv 
		= O_P(1)
		\quad \textnormal{and} \quad
		\left( {T}\inv Z Z' + \sqrt{T} \Lambda \right)\inv \left( \sqrt{T} \Lambda \right) 
		 = O_P(1) ,
	\end{equation*}
	since $\Lambda \in \mathcal{I}_\lambda$. It follows
	\begin{equation*}
		\textnormal{Err}\left( \hat{\bm{\beta}}^R_\blacklozenge(\Lambda) \right)
		= 
		\bm{\beta}' \left( \sqrt{T} \Lambda \right) \left( \frac{Z Z'}{T} + \sqrt{T} \Lambda \right)\inv \Gamma \left( \frac{Z Z'}{T} + \sqrt{T} \Lambda \right)\inv \left( \sqrt{T} \Lambda \right) \bm{\beta} + \Sigma_u + O_P(\sqrt{T}) .
	\end{equation*}
	One can now consider a sequence $\widetilde{\Lambda} = O_p(T^{-1/2})$ of regularizers in $\mathcal{I}_\lambda$. 
	By taking the limit, one gets that
	\begin{align*}
		\lim_{T \to \infty} \textnormal{Err}\left( \hat{\bm{\beta}}^R_\blacklozenge \left( \widetilde{\Lambda} \right) \right)
		& =
		\bm{\beta}' \cdot O_p(1) \cdot \left( \Gamma + O_p(1) \right)\inv \Gamma \left( \Gamma + O_p(1) \right)\inv \cdot O_p(1) \cdot \bm{\beta} + \Sigma_u \\
		& \geq 
		\Sigma_u
		= 
		\lim_{T \to \infty} \textnormal{Err}\left( \hat{\bm{\beta}}^R_\blacklozenge(0) \right) ,
	\end{align*}
	meaning that $\widetilde{\Lambda}$ can not be optimal asymptotically, since the least squares solution at $\Lambda = 0$ achieves a lower predictive error. Additionally, any sequence with lower convergence order is also asymptotically invalid. Therefore, by contradiction, it must hold that $\Lambda_\blacklozenge = o_p(T^{-1/2})$.
\end{proof}

\paragraph{Proof of Corollary~\ref{corollary:asym_valid_cv_shrink}}

\begin{proof}
	The first results follows directly from Theorem~\ref{th:asym_valid_cv_penalty}. Further, the fact that $\Lambda_{2,\blacklozenge} = O_P(1)$ is trivial because it is assumed that $\Lambda \in \mathcal{I}_\lambda$.
	
	To prove the second part of the theorem, one can simply notice that, given the assumption on the coefficients $\bm{\beta}$ in Theorem~\ref{th:MAIN_TEOREM_with_shrink}, 
	\begin{equation*}
		\left( \frac{Z Z'}{T} + \sqrt{T} \Lambda \right)\inv \left( \sqrt{T} \Lambda \right) \bm{\beta}
		=
		\left( \frac{Z Z'}{T} + \sqrt{T} \Lambda \right)\inv 
		\begin{bmatrix}
			\sqrt{T} \Lambda_1 \cdot \bm{\beta}_1 \\
			\Lambda_2 \cdot T^{-\delta} \bm{b}_2
		\end{bmatrix}
		\toprob
		0 ,
	\end{equation*}
	where $\Lambda$ was partitioned into two diagonal blocks, $\Lambda_1$ and $\Lambda_2$, as done previously. Block $\Lambda_1$ must be $O_P(T^{-1/2})$ following the proof of Theorem~\ref{th:asym_valid_cv_penalty}. Finally, the fact that $T^{-\delta} \bm{b}_2 \to 0$ as $T \to \infty$ means that, in the limit, a nonzero $\Lambda_2$ does not yield a sub-optimal cross-validation loss value.  
\end{proof}

\section{Monte Carlo Simulations}\label{section:simulation_details}

\subsection{Cross-validation Details}\label{subsection:cv_penalty_selection}

To select the ridge penalty, I implement the lag-adapted structure and choose the relevant $\lambda_i$'s using block non-dependent cross-validation, c.f. \cite{Burman1994,bergmeirNoteValidityCrossvalidation2018}. I constraint the optimization domain of $\lambda_i$ to be $[0, 10^2]$, without discretization. An issue with cross-validation regards the GLS ridge estimator: the matrices involved can quickly become prohibitively large due to Kronecker products, making CV optimization impractical. To avoid this, I set the penalty for lag-adapted $\hat{\bm{\beta}}^{RGLS}(\Lambda)$ to be the same as that obtained for $\hat{\bm{\beta}}^{R}$ via CV, which means the regularizer is tuned sub-optimally. Nonetheless, an identical choice of $\Lambda$ for both methods can help shed light on the difference in structure between the two estimators.

In contrast to \cite{banburaLargeBayesianVector2010}, I do not tune the shrinkage parameter of the Minnesota BVAR using a mean squared forecasting error (MSFE) criterion: instead, I again use block CV. Since a Minnesota prior can be easily implemented with the use of augmented regression matrices, cross-validation can be much more efficiently implemented than for $\hat{\bm{\beta}}^{RGLS}$. The resulting choice of prior tightness $\lambda^2$ is reasonable because CV, too, estimates the (one step ahead) forecasting risk. Since in this context only the mean of the posterior is used to compute pointwise impulse responses, one can even directly interpret the cross-validated Minnesota BVAR estimator as a refinement of GLS ridge.

\subsection{Penalty Selection in Simulations}

In the simulations of Section \ref{section:monte_carlo_simulations}, an interesting aspect to study is how data-driven penalty selection methods behave. Both average and individual behavior are important, because the former generally gives intuition for the kind of regularization structure that is selected for the model, while the latter is relevant for empirical modeling where estimation can be done only on one sample.

I use Setup A from the Monte Carlo experiments (with 1000 replications) and instead just focus on the behavior of a number of penalty validation techniques. The rational behind the choice is straightforward: Setup A involves models with more lags. Figure \ref{fig:cv1} shows the mean selected penalty parameter $\lambda_i$ for $i \in \{1, \ldots, 10\}$. The methods I compare are: out-of-sample validation (OOS) with a split of 80\% of sample for estimation and 20\% for testing; block cross-validation with 5 (CV 5) or 10 (CV 10) folds; block non-dependent cross validation with 10 folds (BND CV 10). The differences between 5 and 10-fold block CV methods are small, and both largely agree with BND cross-validation apart at lag 2. In contrast, out-of-sample validation appears to select on average much higher penalties at early lags and slightly lower ones at higher lags. 

Regarding the distribution of selected $\lambda_i$ over all replications, one can notice from Figure \ref{fig:cv2} that there is indeed important variation in the individual penalty choices. The evidence is for the specific case of 10-fold block CV, but it appears as a common pattern with other techniques, too. The implications of such variability in lag-adapted penalties are hard to gauge because in any given sample it is not possible to say whether the choice of $\{\lambda_i\}$ is good or bad outside of speculation. A guiding principle might be to compare CV with any "hyperpriors" one might have on $\Lambda$ itself -- like in Bayesian paradigm -- but then parameteric penalty matrices like the one used with the Minnesota prior should be preferred. Indeed, the question of whether a more robust but still general method other than cross--validation can be applied to the time series context is highly relevant.

\begin{figure}
	\centering
	\makebox[\textwidth][c]{
		\begin{subfigure}{0.48\textwidth}
			\includegraphics[width=\textwidth]{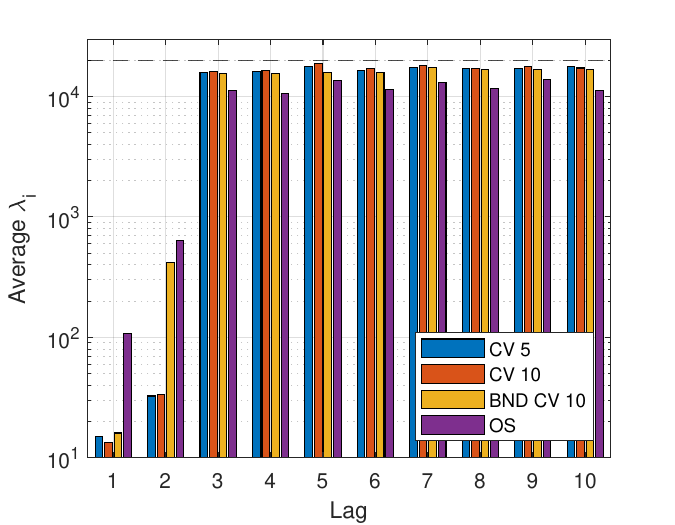}
			\caption{Average penalty selected by lag.}
			\label{fig:cv1}
		\end{subfigure}
		\begin{subfigure}{0.48\textwidth}
			\includegraphics[width=\textwidth]{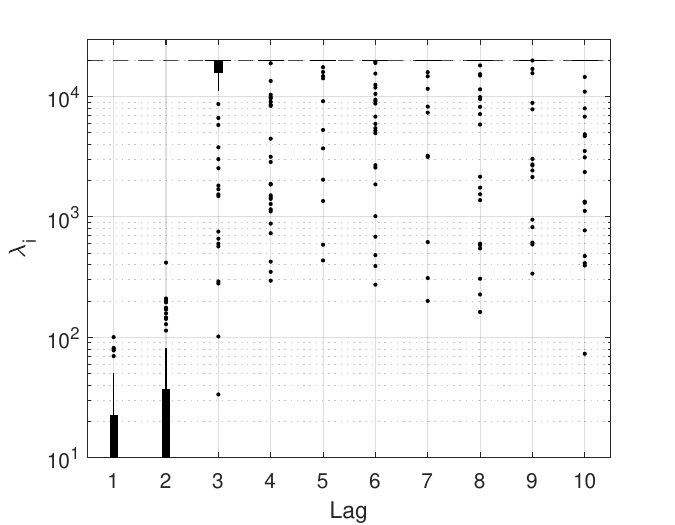}
			\caption{Distribution of 10-fold CV penalties.}
			\label{fig:cv2}
		\end{subfigure}
	}
	\caption{Comparison of penalty selection methods in Setup A, 1000 replications.}
\end{figure}

\subsection{Penalty Selection with Many Lags}

The minimization problem involved in OOS and CV penalty selection in its most general form suffers from the curse of dimensionality. This is somehow mitigated when using a lag-adapted regularizer since the loss with $\Lambda^\ell$ depends only on $p$ non-negative parameters, rather then $K^2 p$ with non-block-diagonal $\Lambda$. But whenever $p$ is chosen large (e.g. $p > 20$) the problem resurfaces.

I suggest a basic shortcut to make computation easier. Such simplification stems from the following observation. If one is willing to believe the assumptions of Section \ref{section:inference_shrinkage}, then, because deep lag coefficients are small, after the first few lags penalization can be equally strong on all remaining lags with negligible additional bias. The shortcut, then, is to estimate only $\{\lambda_1, \ldots, \lambda_r\}$ for $r < p$, then extrapolate and use $\{\lambda_1, \ldots, \lambda_{r-1}, \lambda_r, \ldots , \lambda_r\}$, where $\lambda_r$ is repeated $p-r$ times, as lag-adapted penalty parameters. The idea is supported by the results in Figure \ref{fig:cv1}. However, this strategy is not generally appropriate, because it could be that even at relatively deep lags some coefficients are large, while on the other hand the early coefficients are small. Therefore, in applications where the ridge penalty needs to be estimated only once or a handful of times I would suggest to avoid this shortcut altogether.

\subsection{Numerical Optimization}

For a VAR($p$) and a lag-adapted $\Lambda^\ell$, a collection $\{\lambda_1, \ldots, \lambda_p\}$ must be chosen. To implement OS and CV for the ridge estimators, I rely on MATLAB optimization routines, in both cases using the optimization function \texttt{patternsearch} from the MATLAB Optimization Toolbox. The domain of optimization is chosen to be the hypercube $[0, 10^2]^p$, where $T$ is the sample size. The choice of a bounded domain is asymptotically valid, c.f. Theorem~\ref{th:asym_valid_cv_penalty} and Corollary~\ref{corollary:asym_valid_cv_shrink}.

In applications, since the CV loss needs not be convex \citep{Stephenson2021}, it appropriate to employ advanced optimization routines, e.g. genetic or pattern-based optimizers like \texttt{patternsearch}, if possible. When one only requires to estimate the VAR model once, then the selection of $\Lambda$ is a one-time cost. The gains of better optimization solutions therefore are often superior to the higher computational costs one incurs in when using more sophisticated routines. 

\newpage

\subsection{Additional Tables}\label{appendix:additional_tables}
\vspace{2em}

\begin{table}[h!]
	\centering
	\renewcommand{\arraystretch}{0.9}
	\caption{MSE Relative to OLS -- Setup B}
	\begin{tabular}{lcccccccc}
		\toprule
		Variable & Method & $h$ = 1 & $h$ = 4 & $h$ = 8 & $h$ = 12 & $h$ = 16 & $h$ = 20 & $h$ = 24 \\
		\midrule
		& RIDGE & 1.11 & 1.08 & 1.16 & 1.06 & 0.90 & 0.89 & 0.94 \\
		& RIDGE-GLS & 1.16 & 1.00 & 0.99 & 1.00 & 0.93 & 0.93 & 0.95 \\
		Real GDP & LP & 1.00 & 1.14 & 1.37 & 1.52 & 1.72 & 1.98 & 2.24 \\
		& BVAR-CV & 0.90 & 0.87 & 1.04 & 1.01 & 0.92 & 0.92 & 0.98 \\
		& H-BVAR & 0.83 & 0.62 & 0.78 & 0.73 & 0.62 & 0.62 & 0.68 \\
		\midrule
		& RIDGE & 2.25 & 2.10 & 1.81 & 1.54 & 1.39 & 1.47 & 3.37 \\
		& RIDGE-GLS & 1.17 & 1.09 & 1.10 & 1.10 & 1.08 & 1.05 & 1.04 \\
		GDP Deflator & LP & 1.00 & 1.13 & 1.25 & 1.33 & 1.40 & 1.47 & 1.54 \\
		& BVAR-CV & 1.06 & 1.03 & 1.00 & 0.95 & 0.92 & 0.92 & 0.93 \\
		& H-BVAR & 0.81 & 0.91 & 1.10 & 1.11 & 1.04 & 0.98 & 0.94 \\
		\midrule
		& RIDGE & 0.99 & 1.29 & 1.31 & 1.08 & 0.96 & 1.05 & 2.31 \\
		& RIDGE-GLS & 0.94 & 0.96 & 1.04 & 0.99 & 0.95 & 0.96 & 0.99 \\
		Consumption & LP & 1.00 & 1.13 & 1.32 & 1.44 & 1.63 & 1.83 & 2.03 \\
		& BVAR-CV & 1.04 & 1.07 & 1.15 & 1.00 & 0.92 & 0.94 & 0.99 \\
		& H-BVAR & 0.94 & 0.83 & 0.98 & 0.83 & 0.78 & 0.83 & 0.91 \\
		\midrule
		& RIDGE & 1.49 & 1.27 & 1.17 & 0.99 & 0.70 & 0.73 & 1.61 \\
		& RIDGE-GLS & 1.34 & 1.14 & 1.02 & 1.02 & 0.86 & 0.82 & 0.86 \\
		Investment & LP & 1.00 & 1.15 & 1.40 & 1.63 & 2.03 & 2.76 & 3.59 \\
		& BVAR-CV & 1.51 & 1.01 & 0.97 & 0.97 & 0.93 & 1.08 & 1.24 \\
		& H-BVAR & 1.06 & 0.68 & 0.69 & 0.66 & 0.63 & 0.87 & 1.14 \\
		\midrule
		& RIDGE & 1.22 & 1.24 & 1.18 & 1.03 & 0.77 & 0.76 & 1.27 \\
		& RIDGE-GLS & 1.07 & 1.05 & 1.01 & 1.03 & 0.90 & 0.85 & 0.90 \\
		Hours & LP & 1.00 & 1.14 & 1.33 & 1.53 & 1.81 & 2.35 & 2.92 \\
		& BVAR-CV & 0.89 & 0.88 & 1.03 & 1.02 & 0.91 & 0.95 & 1.05 \\
		& H-BVAR & 0.77 & 0.71 & 0.97 & 0.96 & 0.81 & 0.85 & 0.98 \\
		\midrule
		& RIDGE & 0.85 & 0.99 & 0.85 & 0.94 & 1.13 & 1.40 & 4.70 \\
		& RIDGE-GLS & 0.93 & 0.97 & 0.89 & 0.94 & 1.04 & 1.04 & 1.00 \\
		Compensation & LP & 1.00 & 1.18 & 1.52 & 1.78 & 1.90 & 1.93 & 1.99 \\
		& BVAR-CV & 1.07 & 0.92 & 0.94 & 0.93 & 0.99 & 1.01 & 0.98 \\
		& H-BVAR & 0.86 & 0.80 & 1.12 & 1.31 & 1.35 & 1.27 & 1.22 \\
		\midrule
		& RIDGE & 2.17 & 1.21 & 0.96 & 0.93 & 1.03 & 4.00 & 53.18 \\
		& RIDGE-GLS & 1.21 & 1.04 & 0.90 & 0.93 & 0.90 & 0.88 & 0.91 \\
		Fed Funds & LP & 1.00 & 1.18 & 1.51 & 1.71 & 1.97 & 2.44 & 2.99 \\ 
		Rate & BVAR-CV & 0.92 & 0.94 & 0.91 & 0.90 & 0.86 & 0.87 & 0.92 \\
		& H-BVAR & 0.75 & 0.77 & 1.32 & 1.38 & 1.25 & 1.15 & 1.20 \\
		\bottomrule
	\end{tabular}
	\renewcommand{\arraystretch}{1}
	\label{table:tab:ir_rootmse_2_ext}
\end{table}

\begin{table}
	\centering
	\renewcommand{\arraystretch}{0.9}
	\caption{Impulse Response Inference -- Setup B: CI Coverage}
	\vspace{-3pt}
	\begin{tabular}{ccccccccc}
		\toprule
		Variable & Method & $h$ = 1 & $h$ = 4 & $h$ = 8 & $h$ = 12 & $h$ = 16 & $h$ = 20 & $h$ = 24 \\
		\midrule
		& LS & 0.87 & 0.81 & 0.75 & 0.72 & 0.71 & 0.72 & 0.73 \\
		& RIDGE & 0.90 & 0.79 & 0.66 & 0.62 & 0.65 & 0.68 & 0.68 \\
		Real GDP & RIDGE-AS & 0.89 & 0.72 & 0.61 & 0.58 & 0.61 & 0.65 & 0.65 \\
		& LP & 0.87 & 0.93 & 0.94 & 0.94 & 0.93 & 0.93 & 0.91 \\
		& BVAR-CV & 0.70 & 0.71 & 0.63 & 0.64 & 0.71 & 0.75 & 0.76 \\
		& H-BVAR & 0.84 & 0.86 & 0.76 & 0.76 & 0.83 & 0.88 & 0.88 \\
		\midrule
		& LS & 0.86 & 0.83 & 0.80 & 0.76 & 0.73 & 0.72 & 0.70 \\
		& RIDGE & 0.85 & 0.76 & 0.66 & 0.62 & 0.61 & 0.60 & 0.58 \\
		GDP Deflator & RIDGE-AS & 0.83 & 0.70 & 0.61 & 0.58 & 0.57 & 0.57 & 0.55 \\
		& LP & 0.86 & 0.93 & 0.94 & 0.94 & 0.93 & 0.91 & 0.89 \\
		& BVAR-CV & 0.76 & 0.72 & 0.70 & 0.72 & 0.72 & 0.71 & 0.70 \\
		& H-BVAR & 0.84 & 0.83 & 0.79 & 0.78 & 0.78 & 0.78 & 0.77 \\
		\midrule
		& LS & 0.87 & 0.80 & 0.75 & 0.72 & 0.70 & 0.70 & 0.70 \\
		& RIDGE & 0.90 & 0.74 & 0.60 & 0.60 & 0.64 & 0.66 & 0.65 \\
		Consumption & RIDGE-AS & 0.89 & 0.67 & 0.55 & 0.55 & 0.60 & 0.62 & 0.62 \\
		& LP & 0.86 & 0.93 & 0.94 & 0.94 & 0.94 & 0.92 & 0.90 \\
		& BVAR-CV & 0.73 & 0.66 & 0.60 & 0.63 & 0.70 & 0.73 & 0.74 \\
		& H-BVAR & 0.84 & 0.79 & 0.70 & 0.74 & 0.79 & 0.82 & 0.84 \\
		\midrule
		& LS & 0.87 & 0.82 & 0.76 & 0.73 & 0.75 & 0.82 & 0.87 \\
		& RIDGE & 0.85 & 0.79 & 0.65 & 0.62 & 0.73 & 0.80 & 0.81 \\
		Investment & RIDGE-AS & 0.82 & 0.69 & 0.59 & 0.57 & 0.68 & 0.77 & 0.77 \\
		& LP & 0.87 & 0.94 & 0.94 & 0.95 & 0.94 & 0.94 & 0.94 \\
		& BVAR-CV & 0.70 & 0.73 & 0.67 & 0.71 & 0.77 & 0.81 & 0.83 \\
		& H-BVAR & 0.80 & 0.86 & 0.81 & 0.82 & 0.87 & 0.88 & 0.88 \\
		\midrule
		& LS & 0.86 & 0.81 & 0.76 & 0.74 & 0.74 & 0.79 & 0.81 \\
		& RIDGE & 0.88 & 0.80 & 0.68 & 0.64 & 0.70 & 0.72 & 0.66 \\
		Hours & RIDGE-AS & 0.87 & 0.74 & 0.62 & 0.59 & 0.65 & 0.67 & 0.62 \\
		& LP & 0.86 & 0.93 & 0.94 & 0.94 & 0.94 & 0.94 & 0.93 \\
		& BVAR-CV & 0.73 & 0.72 & 0.64 & 0.66 & 0.74 & 0.78 & 0.77 \\
		& H-BVAR & 0.88 & 0.86 & 0.73 & 0.72 & 0.80 & 0.85 & 0.86 \\
		\midrule
		& LS & 0.86 & 0.82 & 0.76 & 0.75 & 0.72 & 0.68 & 0.67 \\
		& RIDGE & 0.93 & 0.82 & 0.75 & 0.72 & 0.66 & 0.60 & 0.58 \\
		Compensation & RIDGE-AS & 0.91 & 0.71 & 0.69 & 0.67 & 0.61 & 0.56 & 0.55 \\
		& LP & 0.86 & 0.93 & 0.95 & 0.95 & 0.94 & 0.92 & 0.90 \\
		& BVAR-CV & 0.78 & 0.71 & 0.69 & 0.71 & 0.69 & 0.67 & 0.69 \\
		& H-BVAR & 0.85 & 0.83 & 0.80 & 0.81 & 0.82 & 0.82 & 0.83 \\
		\midrule
		& LS & 0.85 & 0.83 & 0.80 & 0.78 & 0.77 & 0.79 & 0.80 \\
		& RIDGE & 0.79 & 0.77 & 0.74 & 0.68 & 0.68 & 0.72 & 0.72 \\
		Fed Funds & RIDGE-AS & 0.78 & 0.66 & 0.68 & 0.64 & 0.64 & 0.68 & 0.69 \\
		Rate & LP & 0.85 & 0.94 & 0.96 & 0.96 & 0.95 & 0.94 & 0.93 \\
		& BVAR-CV & 0.76 & 0.72 & 0.76 & 0.77 & 0.77 & 0.81 & 0.83 \\
		& H-BVAR & 0.87 & 0.86 & 0.74 & 0.73 & 0.78 & 0.84 & 0.87 \\
		\bottomrule
	\end{tabular}
	\renewcommand{\arraystretch}{1}
	\label{table:tab:ir_cover_2_ext}
\end{table}

\begin{table}
	\centering
	\renewcommand{\arraystretch}{0.9}
	\caption{Impulse Response Inference -- Setup B: CI Length (rescaled $\times 100$)}
	\vspace{-3pt}
	\begin{tabular}{ccccccccc}
		\toprule
		Variable & Method & $h$ = 1 & $h$ = 4 & $h$ = 8 & $h$ = 12 & $h$ = 16 & $h$ = 20 & $h$ = 24 \\
		\midrule
		& LS & 0.71 & 1.56 & 2.07 & 2.31 & 2.32 & 2.24 & 2.15 \\
		& RIDGE & 0.79 & 1.56 & 1.85 & 1.95 & 1.92 & 1.85 & 1.77 \\
		Real GDP & RIDGE-AS & 0.74 & 1.31 & 1.65 & 1.76 & 1.75 & 1.70 & 1.64 \\
		& LP & 0.71 & 2.42 & 4.21 & 5.40 & 5.90 & 5.91 & 5.70 \\
		& BVAR-CV & 0.53 & 1.23 & 1.74 & 2.00 & 2.10 & 2.13 & 2.15 \\
		& H-BVAR & 0.58 & 1.36 & 1.87 & 2.16 & 2.32 & 2.44 & 2.55 \\
		\midrule
		& LS & 0.26 & 0.74 & 1.47 & 2.14 & 2.69 & 3.12 & 3.46 \\
		& RIDGE & 0.32 & 0.81 & 1.41 & 1.96 & 2.43 & 2.82 & 3.13 \\
		GDP Deflator & RIDGE-AS & 0.31 & 0.71 & 1.28 & 1.79 & 2.23 & 2.59 & 2.88 \\
		& LP & 0.26 & 1.09 & 2.68 & 4.39 & 5.95 & 7.23 & 8.13 \\
		& BVAR-CV & 0.21 & 0.61 & 1.26 & 1.88 & 2.43 & 2.91 & 3.32 \\
		& H-BVAR & 0.23 & 0.71 & 1.42 & 2.09 & 2.71 & 3.27 & 3.80 \\
		\midrule
		& LS & 0.63 & 1.35 & 1.97 & 2.29 & 2.34 & 2.30 & 2.24 \\
		& RIDGE & 0.71 & 1.35 & 1.74 & 1.93 & 1.96 & 1.93 & 1.88 \\
		Consumption & RIDGE-AS & 0.67 & 1.15 & 1.57 & 1.75 & 1.79 & 1.77 & 1.74 \\
		& LP & 0.63 & 2.04 & 3.92 & 5.27 & 5.89 & 6.00 & 5.88 \\
		& BVAR-CV & 0.49 & 1.09 & 1.65 & 1.97 & 2.12 & 2.20 & 2.24 \\
		& H-BVAR & 0.53 & 1.21 & 1.79 & 2.16 & 2.41 & 2.60 & 2.77 \\
		\midrule
		& LS & 3.38 & 6.65 & 7.89 & 7.89 & 7.31 & 6.69 & 6.18 \\
		& RIDGE & 3.79 & 6.81 & 6.93 & 6.46 & 5.79 & 5.19 & 4.73 \\
		Investment & RIDGE-AS & 3.59 & 5.57 & 6.11 & 5.77 & 5.21 & 4.72 & 4.34 \\
		& LP & 3.37 & 10.16 & 16.00 & 18.85 & 19.06 & 18.22 & 17.23 \\
		& BVAR-CV & 2.64 & 5.26 & 6.59 & 6.91 & 6.78 & 6.57 & 6.38 \\
		& H-BVAR & 2.89 & 5.74 & 7.08 & 7.54 & 7.63 & 7.60 & 7.58 \\
		\midrule
		& LS & 0.70 & 1.64 & 2.27 & 2.42 & 2.29 & 2.11 & 1.99 \\
		& RIDGE & 0.82 & 1.73 & 2.10 & 2.06 & 1.86 & 1.66 & 1.52 \\
		Hours & RIDGE-AS & 0.79 & 1.49 & 1.87 & 1.85 & 1.68 & 1.51 & 1.40 \\
		& LP & 0.70 & 2.49 & 4.52 & 5.62 & 5.83 & 5.57 & 5.26 \\
		& BVAR-CV & 0.57 & 1.30 & 1.90 & 2.07 & 2.03 & 1.94 & 1.88 \\
		& H-BVAR & 0.62 & 1.49 & 2.15 & 2.35 & 2.38 & 2.37 & 2.37 \\
		\midrule
		& LS & 0.86 & 1.17 & 1.18 & 1.18 & 1.21 & 1.24 & 1.25 \\
		& RIDGE & 0.97 & 1.21 & 1.07 & 1.05 & 1.06 & 1.06 & 1.06 \\
		Compensation & RIDGE-AS & 0.93 & 0.95 & 0.95 & 0.96 & 0.97 & 0.98 & 0.98 \\
		& LP & 0.86 & 1.80 & 2.53 & 2.87 & 3.11 & 3.26 & 3.32 \\
		& BVAR-CV & 0.69 & 0.94 & 1.00 & 1.05 & 1.11 & 1.18 & 1.23 \\
		& H-BVAR & 0.78 & 1.10 & 1.26 & 1.40 & 1.54 & 1.67 & 1.78 \\
		\midrule
		& LS & 0.25 & 0.39 & 0.43 & 0.43 & 0.41 & 0.38 & 0.35 \\
		& RIDGE & 0.29 & 0.39 & 0.37 & 0.36 & 0.33 & 0.30 & 0.29 \\
		Fed Funds & RIDGE-AS & 0.27 & 0.31 & 0.33 & 0.32 & 0.30 & 0.28 & 0.27 \\
		Rate & LP & 0.25 & 0.59 & 0.88 & 1.01 & 1.05 & 1.03 & 0.98 \\
		& BVAR-CV & 0.21 & 0.31 & 0.36 & 0.37 & 0.36 & 0.35 & 0.34 \\
		& H-BVAR & 0.23 & 0.36 & 0.42 & 0.44 & 0.45 & 0.45 & 0.46 \\
		\bottomrule
	\end{tabular}
	\renewcommand{\arraystretch}{1}
	\label{table:tab:ir_length_2_ext}
\end{table}


\end{document}